%% file: testing-unate-distributions.tex
\documentclass[11pt]{article}

\usepackage{sn-preamble}
\usepackage{mathtools} % for \coloneqq command

\newcommand{\1}[1]{\mathbf{1}\left[#1\right]}

\newcommand{\Binomial}{\mathsf{Binomial}} % Binomial distribution
\newcommand{\pr}[2]{\Pr_{#1}\left[#2\right]}
\newcommand{\ex}[2]{\Ex_{#1}\left[#2\right]}

\newcommand{\UnifUnate}{\texttt{Unate-Uniformity}}
\newcommand{\bsigma}{\boldsymbol{\sigma}}
\newcommand{\btau}{\boldsymbol{\tau}}
\newcommand{\brho}{\boldsymbol{\rho}}
\newcommand{\dtv}{\mathrm{d}_{\mathrm{TV}}}
\newcommand{\dkl}{\mathrm{d}_{\mathrm{KL}}}
\newcommand{\dchi}{\mathrm{d}_{\mathrm{\chi^2}}}
\newcommand{\SubcubeUnate}{\texttt{Subcube-Unate}}

\title{Testing Unate Distributions}

\author{
    Daeho Lee\\ 
    \emph{Massachusetts Institute of Technology}\\
    \url{daeho0_0@mit.edu}
    \and 
    Shivam Nadimpalli\\
    \emph{Massachusetts Institute of Technology}\\
    \url{shivamn@mit.edu}
    \and 
    Mingda Qiao\\
    \emph{University of Massachusetts, Amherst}\\
    \url{mingda.qiao.cs@gmail.com}
    \and 
    Ronitt Rubinfeld\\
    \emph{Massachusetts Institute of Technology}\\
    \url{ronitt@csail.mit.edu}
    \vspace{0.5em}
}

\date{\today}

\hypersetup{
	colorlinks = true,
    linkcolor={purple},
    urlcolor={purple},
    citecolor={blue!80!black}
}

\begin{document}

\pagenumbering{gobble}

\maketitle

\hypersetup{linkcolor={purple}}

\begin{abstract}
    We initiate the study of \emph{unate distributions} over $\bits^n$---a natural analogue of unate Boolean functions---by considering two basic testing problems that parallel well-studied questions for monotone distributions: 
    \begin{itemize}
	\item \textbf{Uniformity Testing of Unate Distributions:} We show that $\wt{\Theta}(n^{3/2})$ samples are sufficient and necessary, in contrast to the $\wt{\Theta}(n)$ sample complexity of the analogous problem for monotone distributions (Rubinfeld and Servedio, STOC 2005; Adamaszek, Czumaj, and Sohler, SODA 2010). 

	\item \textbf{Unateness Testing of Arbitrary Distributions:} We give a tester that uses $\wt{O}(n^{3/2})$ conditional samples in the subcube conditional model. On the other hand, every tester that draws conditional samples in a similar fashion, namely from $O(1)$-dimensional subcubes, must have an $\wt{\Omega}(n^{2/3})$ complexity. In the same model, the complexity of monotonicity testing was recently shown to be $\wt{\Theta}(n)$~(Chakrabarty et al., STOC 2025).
    \end{itemize}
    Our algorithms for both problems significantly outperform the naive approach of reducing to the monotone case, which would incur $\Omega(n^2)$ sample complexity. 
    Our uniformity tester relies on a subroutine that ``weakly'' learns the hidden orientations of a unate distribution, together with a new correlation bound for these estimates. 
    Both tools may be of independent interest in studying monotonicity and unateness over $\bn$.
\end{abstract}

\newpage

\pagenumbering{arabic}

\input{sections/intro}

\input{sections/prelims}

\input{sections/testing-uniformity}
\input{sections/subcube}

\input{sections/subcube-lower-bound}

\section*{Acknowledgements}

D.L.~is supported by the MIT Undergraduate Research Opportunities Program (UROP). 
S.N.~is partially supported by Elchanan Mossel’s Vannevar Bush Faculty Fellowship ONR-N00014-
20-1-2826. 
R.R.~is supported by the NSF TRIPODS program (award DMS-2022448) and CCF-2310818. 

\appendix

\bibliography{allrefs}
\bibliographystyle{alphaurl}

\end{document}

%% file: sections/intro.tex
\section{Introduction}
\label{sec:intro}

Many fundamental distribution testing tasks---uniformity testing being the canonical example---become infeasible in high dimensional settings such as the Boolean hypercube $\bn$, requiring exponentially many samples in the dimension~\cite{canonne2020survey,canonne2022topics}. 
A central theme in the field has therefore been to identify \emph{structured} settings that admit efficient testers: examples include Bayesian networks~\cite{canonne2017testing,daskalakis2017square}, Markov random fields~\cite{daskalakis2019testing,bezakova2020lower}, and structured truncations~\cite{de2024detecting,de2025testing}. 

Among such structural assumptions, one of the most extensively studied is \emph{monotonicity}, the distributional analogue of the classical Boolean function property.\footnote{Formally, a distribution $\calD$ over $\bn$ is \emph{monotone} if its probability mass function $\calD : \bn\to[0,1]$ is monotone, i.e., $\calD(x) \leq \calD(y)$ whenever $x_i \leq y_i$ for all $i\in[n]$.} 
In the Boolean function setting, monotonicity has played a central role in sublinear algorithms for nearly three decades: it was among the first properties studied in Boolean function testing and continues to be actively investigated to this day~\cite{BshoutyTamon:96,GGLRS,chakrabarty2013n,chen2014new,khot2018monotonicity,lange2022properly,chen2024mildly,black2025monotonicity,lange2025agnostic,chen2025boolean}. 
In the distributional setting, monotonicity has similarly been studied in depth, leading to tight upper and lower bounds for tasks such as uniformity testing and property testing under various access models~\cite{batu2004sublinear,RubinfeldS09,adamaszek2010testing,bhattacharyya2011testing,acharya2015optimal,aliakbarpour2019towards,rubinfeld2020monotone,blanc2023lifting,CCRSW25}. 

A closely related property in the Boolean function setting is \emph{unateness}. 
Informally, a Boolean function is unate if each coordinate is either always non-decreasing or always non-increasing.  
% A Boolean function $f:\bits^n \to \bits$ is unate if $g(x) := f(x \oplus r)$ is monotone for some $r \in \bits^n$, where $\oplus$ denotes coordinate-wise multiplication; equivalently, $f$ is either non-decreasing or non-increasing in each coordinate. 
The problem of testing unateness of Boolean functions was introduced alongside monotonicity testing in~\cite{GGLRS}, and has since received considerable attention~\cite{khot2016n,chakrabarty2016widetilde,CWX17stoc,CWX17focs,chen2019testing,baleshzar2020optimal}. 
Both monotonicity and unateness are fundamental structural properties that arise naturally across many domains, for example in social choice (where we can view a Boolean function $f$ as a voting rule), economics, learning theory, and circuit design. 

Despite extensive work on unate functions, the analogous notion for distributions has, to the best of our knowledge, not been studied. 
In this work, we initiate the study of \emph{unate distributions} over $\bn$. 
The motivation for this is twofold:
\begin{itemize}
    \item The study of monotone distributions has a well-developed literature precisely because monotonicity is a natural structural constraint arising in settings ranging from social choice to learning theory. Unateness is a strict and natural generalization---it captures the same ``directional'' structure without fixing a canonical orientation---and we believe its study for distributions will prove equally fruitful. 
    \item Additionally, unate distributions, which we define shortly, form a broad and natural family of probability distributions. For example, every unate function $f$ induces a unate distribution (the uniform distribution over $f^{-1}(1))$, and the class of unate functions includes families such as halfspaces and read-once decision lists, which need not be monotone. Beyond this connection to functions, every product distribution over $\bn$ is unate.
\end{itemize}

Finally, as this work demonstrates, unate distributions pose challenges that cannot be addressed by simply reducing to the monotone setting. 
Even for basic tasks such as uniformity testing, the sample complexity changes, and more broadly, new ideas are required to design and analyze optimal testers for this richer class of distributions. 

\subsection{Our Results}
\label{subsec:our-results}

We start by formally defining unate distributions over $\bits^n$: 

\begin{definition} \label{def:unate-dist}
	A distribution $\calD$ over $\bits^n$ is \emph{unate} if there exists $\sigma \in \bits^n$ such that $\calD(x\oplus \sigma)$ is a monotone probability mass function, where $\oplus$ denotes coordinate-wise multiplication. 
\end{definition}

Unate distributions immediately inherit many of the algorithmic challenges and lower bounds known from the monotone distribution setting. 
This includes, for example, the exponential sample lower bounds for tasks such as entropy estimation and independence testing~\cite{RubinfeldS09}. 
In this work, we give efficient algorithms as well as lower bounds for two basic problems:
\begin{itemize}
    \item Uniformity testing of unate distributions; and 
    \item Testing unateness of an unknown distribution in the \emph{subcube conditional model} introduced by Canonne, Ron, and Servedio~\cite{canonne2015testing} and Chakraborty, Fischer, Goldhirsh, and Matsliah~\cite{chakraborty2016power}. 
\end{itemize}

In the subcube conditional model for distributions over $\bn$, an algorithm can query $\calD$ with a subcube $\rho \in \{\pm1, \ast\}^n$ and receive an independent sample $\bx\sim\calD$ conditioned on $\bx_i = \rho_i$ for all $i$ where $\rho_i \neq \ast$. 
This model provides a distributional analogue of the ``membership query'' from learning theory and property testing of Boolean functions, and has received considerable attention in recent years.

We note that for both of the problems we consider, an easy argument using the Chernoff bound together with a union bound over all $\sigma \in \bn$ (cf.~\Cref{def:unate-dist}) allow us to reduce to the corresponding problems for monotone distributions with a multiplicative $O(n)$ sample/subcube query overhead. 
Our main algorithmic results improve on this naive approach for each task; we defer a technical overview of our results to~\Cref{sec:tech-overview} and survey related work in~\Cref{subsec:related-work}. 

\subsubsection{Uniformity Testing of Unate Distributions}

Throughout, we write $\calU = \calU_n$ for the uniform distribution over $\bn$. 
Our first result gives an efficient uniformity tester for high-dimensional distributions under the promise of unateness: 

\begin{restatable}[]{theorem}{unifUB}
\label{thm:our-unif-test}
    Let $\eps \in (0,1)$. There is an algorithm, \emph{\UnifUnate} (\Cref{alg:unif-unate}) which, given i.i.d.~sample access to an unknown unate distribution $\calD$ over $\bn$, draws $\wt{O}(n^{3/2}/\eps^2)$ samples and has the following performance guarantee:
    \begin{itemize}
        \item If $\calD = \calU$, then it outputs ``accept'' with probability $9/10$. 
        \item If $\dtv(\calD, \calU) \geq \eps$, then it outputs ``reject'' with probability $9/10$. 
    \end{itemize}
    Furthermore, it runs in $\wt{O}(n^{5/2}/\eps^2)$ time.
\end{restatable}

Here $\dtv(\cdot)$ refers to the total variation or statistical distance (\Cref{subsec:tv-kl-etc}). 
This result should be contrasted with the $\wt{O}(n/\eps^2)$ bound for uniformity testing of monotone distributions due to Rubinfeld and Servedio~\cite{RubinfeldS09}.\footnote{See also Adamaszek, Czumaj, and Sohler~\cite{adamaszek2010testing} who improved the sample complexity to $O(n/\eps^2)$.} 
At a high level, moving from monotone to unate distributions introduces an unknown orientation vector $\sigma$ (cf.~\Cref{def:unate-dist}); any naive strategy that first learns $\sigma$ and then runs the \cite{RubinfeldS09} tester provably incurs an $\Omega(n^2)$ sample overhead.
Our algorithm bypasses this bottleneck by never learning $\sigma$ explicitly. Instead, we introduce a weak-orientation learning framework whose errors are provably weakly correlated across coordinates. 
This allows a small average bias to be amplified into a global signal using only $\widetilde{O}(n^{3/2})$ samples. 
We believe this structural result on limited correlations among marginals of unate distributions is of independent interest. 
% A direct attempt to adapt the \cite{RubinfeldS09} tester to the unate setting---by learning the unknown orientation vector $\sigma$ (cf.~\Cref{def:unate-dist})---provably requires $\Omega(n^2)$ samples. 
% Our $\wt{O}(n^{3/2})$ sample tester circumvents this barrier through a more delicate analysis that exploits limited correlations among the marginals of unate distributions. 
We complement~\Cref{thm:our-unif-test} with a matching lower bound: 

\begin{restatable}[]{theorem}{unifLB}
\label{thm:our-unif-lb}
    Let $\calA$ be any algorithm which, given i.i.d.~sample access to an unknown distribution $\calD$, has the performance guarantee from~\Cref{thm:our-unif-test} with $\eps = 1-n^{-10}$. 
    Then $\calA$ must draw $\Omega(n^{3/2}/\log^2 n)$ samples from $\calD$. 
\end{restatable}

Our proof of \Cref{thm:our-unif-lb} extends the ``monotone decomposition method'' of Rubinfeld and Servedio~\cite{RubinfeldS09} to the setting of unate distributions; see \Cref{sec:tech-overview} for more details.
 
\begin{remark} 
One might hope that the additional power of the subcube conditional model could yield improved bounds for uniformity testing of unate distributions. 
Recall that uniformity testing of \emph{arbitrary} distributions can be done using $\wt{O}(\sqrt{n})$ subcube queries~\cite{canonne2021random}. 
Unfortunately, it turns out that the additional structure of unateness cannot beat this baseline: Chakrabarty, Chen, Ristic, Seshadhri, and Waingarten~\cite{CCRSW25} have recently shown that even the easier problem of testing uniformity of monotone distributions in the subcube model already requires $\wt{\Omega}(\sqrt{n})$ queries.  
\end{remark}

\subsubsection{Testing Unateness of Distributions}

We now turn to the problem of testing whether an unknown distribution over $\bn$ is itself unate. 
In the standard access model where one receives i.i.d.~samples from the unknown distribution being tester, this task is intractable: even for the simpler case of monotonicity, classical ``birthday paradox'' arguments imply that exponential-in-$n$ sample complexity is unavoidable~\cite{aliakbarpour2019towards,rubinfeld2020monotone}. 
It is therefore necessary to consider stronger access models to obtain meaningful algorithms. 

Motivated by this, Chakrabarty, Chen, Ristic, Seshadhri, and Waingarten~\cite{CCRSW25} studied monotonicity testing in the {subcube conditional model}, where the tester can query samples conditioned on arbitrary coordinate restrictions. 
Building on their results, we give an efficient algorithm  for testing unateness in the same model. 

\begin{restatable}[]{theorem}{subcubeUB}
\label{thm:our-subcube-UB}
    Let $\eps \in (0,1)$. 
    There is an algorithm, \emph{\SubcubeUnate} (\Cref{alg:subcube-unate}) which, given subcube query access to an unknown distribution $\calD$ over $\bn$, makes $\wt{O}(n^{3/2}/\eps^2)$ queries and has the following performance guarantee:
    \begin{itemize}
        \item If $\calD$ is unate, it outputs ``accept'' with probability at least $2/3$. 
        \item If $\calD$ is $\eps$-far in TV~distance from every unate distribution on $\bn$, then it outputs ``reject'' with probability at least $2/3$. 
    \end{itemize}
\end{restatable}

We note that the algorithm of~\Cref{thm:our-subcube-UB} also works in the weaker \emph{coordinate oracle} model of Blanca et al.~\cite{blanca2024complexity} where queries are only made on one-dimensional subcubes. 
We complement \Cref{thm:our-subcube-UB} with the following lower bound: 

\begin{restatable}[]{theorem}{subcubeLB}
\label{thm:subcube-lower-bound}
    There exists a constant $\eps_0 > 0$ such that the following holds: If an algorithm tests whether an unknown distribution $\calD$ over $\bits^n$ is unate or $\eps_0$-far from unate by drawing $Q_1$ samples from $\calD$ and making $Q_2$ (possibly adaptive) subcube queries on subcubes of dimensions $\le d$, we have
    \[
        Q_1 + 2^dQ_2 \ge \Omega(n^{2/3} / \log^3 n).
    \]
\end{restatable} 

We note that all prior lower bounds in the subcube model are based on product distributions as hard instances. 
Since every product distribution is trivially unate, these techniques cannot be applied to this setting.
To the best of our knowledge, ours is the first lower bound that breaks this barrier; we return to this point in~\Cref{subsec:discussion}. 

\subsection{Main Ideas} 
\label{sec:tech-overview}

We give a brief technical overview of our results. 
Although both of our upper bounds end up with the same asymptotic sample complexity of $\wt{O}(n^{3/2}/\varepsilon^2)$, that the ideas underlying each of them are quite different. 

\subsubsection{Testing Uniformity of Unate Distributions}

\paragraph{Upper Bound.} 

A natural first attempt is to adapt the Rubinfeld--Servedio~\cite{RubinfeldS09} uniformity tester for monotone distributions to the unate setting.  
The idea would be to first learn the hidden ``orientation vector'' $\sigma^\ast \in \bits^n$ (cf.~\Cref{def:unate-dist}) such that the reoriented distribution $\calD^{\sigma^\ast}$, where $\bx\sim\calD^{\sigma^\ast}$ is obtained by drawing $\by \sim \calD$ and setting $\bx_i = \sigma^\ast_i \cdot \by_i$, is monotone. 
Since total variation distance is preserved under bijections, $\dtv(\calD, \calU) = \dtv(\calD^{\sigma^\ast}, \calU)$. 
Thus, if $\sigma^\ast$ were known, one could directly apply the $O(n/\varepsilon^2)$-sample uniformity tester of~\cite{RubinfeldS09,adamaszek2010testing}. 

However, exactly recovering $\sigma^\ast$ can be prohibitively expensive. 
Consider, for example, a product distribution where each coordinate has mean $\pm \Theta(1/n)$. 
Such a distribution is $\Omega(1)$-far from uniform, but distinguishing $\E[x_i]=0$ from $\E[x_i]=\pm \Theta(1/n)$ for each coordinate requires $\Omega(n^2)$ samples by Chernoff bounds. 
This naive approach of fully learning the orientation thus fails to give the desired sample complexity of $\wt{O}(n^{3/2})$. 

Our algorithm circumvents this barrier by avoiding the need to learn $\sigma^\ast$ exactly. 
Instead, we estimate a weakly correlated proxy for $\sigma^\ast$ using only $\wt{O}(n^{3/2}/\varepsilon^2)$ samples. 
The key technical step (\Cref{prop:BE-bash}; which may be of independent interest) shows that for unate distributions, the coordinate-wise sign estimates cannot exhibit strong correlations: indeed, their covariances decay at rate $O(1/\sqrt{m})$ after $m$ samples. 
This weak dependence is sufficient to amplify a small average bias across coordinates into a global signal, allowing us to test uniformity with $\wt{O}(n^{3/2})$  samples. 
\paragraph{Lower Bound.} 

Our lower bound builds on the ``monotone decomposition'' method of Rubinfeld and
Servedio~\cite{RubinfeldS09}, which we briefly recall. 
Concretely, one starts from a (slightly biased) monotone product distribution $\calD_\tau$, which admits a representation as a convex combination of
uniform distributions on monotone subcubes (the ``decomposition''). 
Using this representation, we define a random distribution $\bq$ by sampling $T$ such subcubes
and averaging their uniform measures, and show that for appropriate parameters $\bq$ is a monotone distribution that is (with high probability) very far from uniform even though it is indistinguishable from $\calD_\tau$ with $\ll n$ samples. 

To lift this construction from monotone to unate distributions, we apply a
\emph{random orientation}: we draw a uniformly random $\bsigma\in\{\pm1\}^n$ and
output $\bq' := \bq^{\bsigma}$ which is the distribution $\bq$ but with output bits flipped according to $\bsigma$ (see \Cref{notation:D-sigma} for a precise definition). 
This preserves the distance from uniformity since total variation distance is invariant under bijections, while hiding the global direction of
bias and ensuring that $\bq'$ is a valid hard instance in the unate class. 
The key point is that although $\bq'$ is far from uniform with high probability, it
is information-theoretically difficult to distinguish $\bq'$ from $U$ using fewer
than $\widetilde{\Omega}(n^{3/2})$ samples. The proof proceeds in two steps:
\begin{itemize}
    \item First, a coupling argument shows that $m$ samples from $\bq'$ are close
(in total variation distance) to $m$ samples from the randomly oriented biased product
distribution $\calD_\tau^{\bsigma}$. (We note that this step holds for every $m$.)
    \item Second, averaging over the random orientation $\bsigma$
forces strong cancellations, and a $\chi^2$-divergence calculation on one-dimensional
marginals shows that $(\calD_\tau^{\bsigma})^{\otimes m}$ itself is close to $U^{\otimes m}$
whenever $m \ll n^{3/2}$.
\end{itemize}
Putting these together yields our matching lower bound,
up to polylogarithmic factors.

% Our lower bound relies on the ``monotone decomposition'' method of Rubinfeld and Servedio~\cite{RubinfeldS09}. \red{TODO}. By constructing mixtures of biased product distributions and then randomly flipping coordinates, we obtain hard instances of unate distributions that are far from uniform, but look nearly uniform to any tester using fewer than $\wt{\Omega}(n^{3/2})$ samples. 
% This yields our matching lower bound (up to polylogarithmic factors). 

\subsubsection{Testing Unate Distributions via Subcube Queries}
\label{subsec:subcube-overview}

\paragraph{Upper Bound.} For the task of testing whether an unknown distribution is itself unate, we work in the \emph{subcube conditional model}, following the recent work of Chakrabarty et al.~\cite{CCRSW25}. 
At a high level, our algorithm adapts the classical \emph{edge tester} for Boolean unateness~\cite{GGLRS} to the distributional setting. The central tool is a structural lemma of~\cite{CCRSW25}, based on a real-valued ``directed'' generalization of an isoperimetric inequality due to Talagrand~\cite{Talagrand:93}, which controls the ``edge bias'' of monotone distributions. 
We refer the reader to~\Cref{lemma:CCRSW} for a precise statement.  
Using this lemma, we show that if $\calD$ is $\varepsilon$-far from unate, then there exists some coordinate $i$ that simultaneously witnesses both ``monotone'' and ``anti-monotone'' violations. 
This can then be detected by making $\wt{O}(n^{3/2})$ subcube edge queries; the analysis requires an elementary but careful calculation. 

\paragraph{Lower Bound.} 

On the lower-bound side, we reduce Boolean unateness testing to distributional unateness. 
Given a Boolean function
$f:\bn\to\{0,1\}$, let $D_f$ be the uniform distribution on $f^{-1}(1) \sse \bn$. 
Our reduction is based on two key facts: 
\begin{itemize}
    \item First, we show that (i) if $f:\bn\to\zo$ is unate, then so is $D_f$, and (ii) conversely, if 
$f$ is $\eps$-far from unate (in Hamming distance), $D_f$ is $\Omega(\eps)$-far from unate (in total variation distance). 
    \item Second, we note that membership query access to $f$ allows us to simulate unconditional samples from $D_f$ and samples from $D_f$ restricted to $d$-dimensional
subcubes (with $2^d$ queries per restriction). 
\end{itemize}
Combining these observations with the $\wt{\Omega}(n^{2/3})$ query lower bound for Boolean unateness testing~\cite{CWX17stoc} yields~\Cref{thm:subcube-lower-bound}. 

\subsection{Related Work}
\label{subsec:related-work}

Monotonicity of Boolean functions is among the most extensively studied properties in sublinear algorithms; we will not attempt to survey the very large body of results here and instead refer the reader to the comprehensive discussion from~\cite{CCRSW25}. 
The problem of unateness testing of Boolean functions was introduced alongside that of monotonicity testing in~\cite{GGLRS}, and a sequence of works~\cite{GGLRS,khot2016n,chakrabarty2016widetilde,CWX17stoc,CWX17focs,chen2019testing} have pinned down its exact query complexity: $\wt{\Theta}(n^{2/3})$ queries are both necessary and sufficient to test unateness of Boolean functions. 

A number of natural distribution testing tasks have exponential sample complexity in high-dimensional settings. 
This includes uniformity testing over $\bits^n$; see the surveys~\cite{canonne2020survey,canonne2022topics} for more discussion on this. 
In order to overcome these lower bounds, many works either make structural assumptions on the distribution being tested, or assume stronger access models to the distribution. 
Examples of the former include monotonicity~\cite{RubinfeldS09,adamaszek2010testing}, Bayesian networks~\cite{canonne2017testing,daskalakis2017square}, Markov random fields~\cite{daskalakis2019testing,bezakova2020lower}, various classes of structured truncations~\cite{he2023testing,de2024detecting,de2025testing}; see also Section~7 of~\cite{canonne2020survey}. 

The subcube conditioning model, introduced by~\cite{canonne2015testing,chakraborty2016power,bhattacharyya2018property}, takes the other route of assuming stronger access to the distribution being tested. 
This model has received much attention in recent years~\cite{canonne2021random,chen2021learning,narayanan2021tolerant,chen2024uniformity,adar2024improved,blanca2024complexity,CCRSW25}, and has also found applications beyond distribution testing to problems in learning theory~\cite{chen2021learning,blanc2023lifting}.
There has also been a nascent line of work on using the subcube conditioning model to make high-dimensional distribution testing practical~\cite{meel2020testing,pote2022scalable}. 

The works most relevant to our results are \cite{RubinfeldS09} and \cite{CCRSW25}. 
Rubinfeld and Servedio~\cite{RubinfeldS09} establish an essentially tight bound of $\wt{\Theta}(n)$ samples for the problem of testing uniformity of an unknown monotone distribution, whereas Charkabarty et al.~\cite{CCRSW25} give $\wt{\Theta}(n)$-query upper and lower bounds for testing monotonicity of an unknown distribution in the subcube model. 
Both our results---for uniformity testing as well as unateness testing---rely on ingredients going into the analyses of~\cite{RubinfeldS09,CCRSW25} respectively. 

\subsection{Discussion}
\label{subsec:discussion}

Our results raise a number of intriguing questions for further study; we highlight two compelling directions below. 

\paragraph{Lower Bounds for Unateness Testing with Subcube Queries.} 

Recall that \Cref{thm:subcube-lower-bound} gives a $\wt{\Omega}(n^{2/3})$ query lower bound against algorithms that make $O(1)$-dimensional subcube queries. 
At present, all known lower-bounds in the subcube conditional model~\cite{canonne2021random,chen2021learning,CCRSW25} construct hard instances that are product distributions, which are amenable to moment-matching techniques~\cite{chen2021learning,CCRSW25}. 
However, since product distributions are automatically unate, none of these techniques can be lifted the unateness setting. 
Improving on~\Cref{thm:subcube-lower-bound} is thus likely to lead to new insights and techniques for proving lower bounds in the subcube conditional model. 

\paragraph{$k$-Monotone Distributions.} 

Finally, one can consider distributional analogues of other structural generalizations of monotonicity.  
A natural candidate is the class of \emph{$k$-monotone distributions}, where the distribution's mass function changes direction at most $k$ times along any monotone path from $-1^n$ to $1^n$.  
The class of $k$-monotone functions has already been studied in the Boolean function setting~\cite{grigorescu2017k,canonne2019testing,chen2022new}, and $k$-monotone distributions could provide another structured setting for distribution testing beyond monotonicity.

%% file: sections/prelims.tex
\section{Preliminaries}
\label{sec:prelims}

We use boldfaced letters such as (e.g.~$\bx \sim \bits^n$) to denote random variables.
Unless explicitly stated otherwise, all probabilities and expectations will be with respect to the uniform distribution. 
Throughout, we will write $\calU_n$ for the uniform distribution over $\bn$ and when the dimension $n$ is clear from context we will simply write $\calU$ instead. 

For a distribution $\calD$, it will be convenient to write $\calD^{\otimes m}$ for the distribution of $m$ independent draws from $\calD$. 
We will use the following notation throughout: 

\begin{notation}
\label{notation:D-sigma}
    Given a distribution $\calD$ over $\bn$ and $\sigma \in \bn$, we write $\calD^\sigma$ for the distribution over $\bn$ where a draw $\bx\sim\calD^\sigma$ is obtained by first drawing $\by\sim\calD$ and then setting $\bx_i \coloneqq \sigma_i\cdot\by_i$ for $i\in[n]$. 
\end{notation}

\subsection{The Berry--Esseen Theorem}
\label{subsec:BE}

We will write $N(0, I_d)$ for the $d$-dimensional standard Gaussian distribution, where $I_d$ denotes the $d\times d$ identity matrix. 
We recall the (multi-dimensional) Berry--Esseen theorem (which will be used to prove~\Cref{prop:BE-bash}):

\begin{theorem}[Theorem~1.1 of~\cite{Bentkus:03}] \label{thm:BE}
    Let $\bX_1, \dots, \bX_m$ be i.i.d.~random vectors in $\R^d$ distributed as $\bX$ where $\E[\bX] = 0$ and $\E[\bX\bX^\top] = I_d$. 
    Then for any convex set $A \sse \R^d$, we have 
    \[
        \abs{\Prx\sbra{\frac{\sum_{i=1}^m \bX_i}{\sqrt{m}}\in A} - \Prx[\bG \in A]} \leq O(1)\cdot \frac{d^{1/4}}{\sqrt{m}}\cdot\Ex\sbra{\|\bX\|^3}\,.
    \]
    where $\bG\sim N(0,I_d)$. 
\end{theorem}

\subsection{Distance Metrics Between Probability Distributions} 
\label{subsec:tv-kl-etc}

We will identify a distribution $\calD$ over a discrete domain $\Omega$ with its density function $\calD : \Omega \to [0,1]$. 
Recall that for two distributions $\calD_1$, $\calD_2$ over a discrete domain $\Omega$, the \emph{total variation} (or \emph{statistical}) distance between $\calD_1$ and $\calD_2$ is 
\[
    \dtv(\calD_1,\calD_2) = \frac{1}{2}\sum_{x\in\Omega} \abs{\calD_1(x) - \calD_2(x)}\,.
\]
We will frequently make use of the fact that total variation distance is invariant under bijections. 

We also recall the 
\emph{Kullback--Leibler (KL)} and 
$\chi^2$-divergences: 
\[
    \dkl(\calD_1,\calD_2) = \sum_{x\in\Omega} \calD_1(x)\log\frac{\calD_1(x)}{\calD_2(x)}
    \qquad\text{and}\qquad 
    \dchi(\calD_1,\calD_2) = \sum_{x\in\Omega} \frac{\pbra{\calD_1(x) - \calD_2(x)}^2}{\calD_2(x)}
\]
respectively. 
The following relationships are standard: 
\begin{equation} \label{eq:distance-relationships}
    \dtv(\calD_1, \calD_2) 
    \leq \sqrt{\frac{1}{2}\dkl(\calD_1,\calD_2)}
    \leq \sqrt{\frac{1}{2}\log\pbra{1 + \dchi(\calD_1,\calD_2)}}\,.
\end{equation}

We will use the following version of the data processing inequality: 

\begin{fact} \label{fact:dpi}
    Let $\bX_1$, $\bX_2$ be random variables over the same domain. 
    For any (possibly randomized) algorithm $\calA$, we have $\dtv\pbra{\calA(\bX_1), \calA(\bX_2)} \leq \dtv(\bX_1, \bX_2)$. 
\end{fact}

%% file: sections/testing-uniformity.tex
\section{Testing Uniformity of Unate Distributions}
\label{sec:uniformity}

We will prove~\Cref{thm:our-unif-test,thm:our-unif-lb} in this section, starting with the former. 

\subsection{Upper Bound}
\label{subsec:unif-ub}

We first recall the principal technical lemma of \cite{RubinfeldS09}'s uniformity tester for monotone distributions:

\begin{lemma}[Corollary~6 of~\cite{RubinfeldS09}] 
\label{lemma:RS}
	If $\calD$ is a monotone distribution over $\bn$ that is $\eps$-far from uniform, then 
	\[
		\Ex_{\bx\sim\calD}\sbra{\sumi \bx_i} \geq \eps\,.
	\]
\end{lemma}

% We will also require the following lemma due to~\cite{adamaszek2010testing}: 
% \begin{lemma}[Lemma~3.2 of~\cite{adamaszek2010testing}]
% \label{lemma:adamszek}
% 	Suppose $\calD$ is a monotone distribution over $\bn$ that is $\eps$-far from uniform. 
% 	If $m \geq BLAH$ and we sample $s$ points $\x{1},\dots, \x{s} \sim\calD$, then 
% 	\[
% 		\Prx\sbra{\sum_{\ell=1}^s \sumi \x{\ell}_i \leq \frac{s\eps}{2}} \leq \frac{12}{13}\,. 
% 	\]
% \end{lemma}

We now turn to our uniformity tester for unate distributions. 
First, note that if $\calD$ is a unate distribution, then there exists some $\sigma^\ast \in \bn$ such that $\calD^{\sigma^\ast}$ is a monotone distribution. In particular, note that we may take $\sigma^\ast_i = \sign(\Ex_{\bx \sim \calD}[\bx_i])$. 
This suggests a natural algorithm: estimate $\sigma^\ast$ by drawing $m$ samples $\x{1}, \x{2}, \ldots, \x{m}$ from $\calD$ and setting 
\[
	{\bsigma}_i := \sign\pbra{\sum_{j=1}^m \x{j}_i}\,,
\]
and then run the \cite{RubinfeldS09} uniformity tester on the monotone distribution $\calD^{{\bsigma}}$. 
Indeed, since
\[
	\dtv(\calD^{{\bsigma}}, \calU) = \dtv(\calD^{{\bsigma}}, \calU^{\bsigma}) = \dtv(\calD, \calU),% \geq \eps \,.
\]
replacing $\calD$ with $\calD^{\bsigma}$ preserves the uniformity (or distance from uniformity) of the distribution.

However, learning the whole sign vector $\sigma^\ast$ can be too expensive. 
In particular, while \Cref{lemma:RS} ensures that $\sumi \Ex[\bx_i]$ is large, it gives no guarantee on the magnitude of each individual $\Ex[\bx_i]$. 
In particular, note that distinguishing $\Ex[\bx_i]=0$ (which corresponds to the uniform distribution) from $\Ex[\bx_i]=\pm\Theta(1/n)$ requires $\Theta(n^2)$ samples per coordinate by standard Chernoff lower bounds.
We bypass this obstacle by not insisting on the exact orientation and instead \emph{weakly} learning the sign vector $\sigma^\ast$.

\subsubsection{Warm-Up: A Sub-Quadratic Tester in a Special Case} 
\label{subsec:subquad-unif-test}

As a warm-up, we show how to obtain a sub-quadratic tester in the hard instance just described above, where naively learning $\sigma^\ast$ requires $\Omega(n^2)$ queries.
It will be convenient to write 
\[
	\mu_i \coloneqq \Ex_{\bx\sim\calD}[\bx_i]\,. 
\]
We make one simplifying assumption for now: $\mu_i = \frac{1}{n}$ holds for every $i \in [n]$, i.e., the unate distribution $\calD$ is in fact monotone, and each coordinate has the same bias of $\Theta(1/n)$ towards $+1$. 

Let $p > 0$ be a parameter that we will set later. 
By the assumption that $\mu_i = \frac{1}{n}$, it can be verified that $m = \Theta(n^2p^2)$ samples suffice for the ``orientation step'' to produce random signs $\bsigma_1, \dots, \bsigma_n \in \bits$ with the promise that 
\[
	\Ex[\bsigma_i] \geq p~\text{for all}~i\in[n]\,.
\]
Equivalently, we correctly deduce that $\sigma^*_i = +1$ with probability $1/2 + \Omega(p)$, which is slightly better than a random guess.

Writing $\bX \coloneqq \frac{1}{n} \sumi \bsigma_i \in [-1,1]$, note that $\E[\bX] \geq p$. 
Applying Markov's inequality to the non-negative random variable $1-\bX$ gives 
\[
	\Prx\sbra{\bX \leq \frac{p}{2}} = \Pr\sbra{1 - \bX \geq 1-\frac{p}{2}} \leq \frac{\Ex\sbra{1-\bX}}{1-\frac{p}{2}} \leq \frac{1-p}{1-\frac{p}{2}} \leq 1-\Omega(p)\,.
\]	

In other words, the event $\frac{1}{n} \sumi \bsigma_i \geq \frac{p}{2}$ holds with probability $\Omega(p)$. 
Repeating this ``orientation step'' independently $O(1/p)$ times ensures that with constant probability, at least one draw satisfies $\bX \geq p/2$.
Then, we note that, conditioning on the realization of $\bsigma \in \bits^n$, the $\bsigma$-weighted Hamming weight $\sum_{i=1}^{n}\bsigma_i\bx_i$ for $\bx \sim \calD$ has an expectation of exactly
\[
    \sum_{i=1}^{n}\bsigma_i \cdot \ex{\bx \sim \calD}{\bx_i}
=   \sum_{i=1}^{n}\bsigma_i \cdot \mu_i
=   \frac{1}{n}\sum_{i=1}^{n}\bsigma_i.
\]
Therefore, by estimating $\frac{1}{n}\sum_{i=1}^{n}\bsigma_i$ to an additive error of $O(p)$ using $\widetilde O(n/p^2)$ additional samples, we obtain a tester with sample complexity (modulo polylogarithmic factors)
\begin{equation} \label{eq:stitch}
	\frac{1}{p}\cdot n^2p^2 + \frac{n}{p^2}.
\end{equation}
Setting $p = n^{-1/3}$ gives a sample complexity of $\widetilde O(n^{5/3})$, which already improves upon the easy $\widetilde O(n^2)$ bound (cf.~\Cref{subsec:our-results}).

\subsubsection{Main Technical Lemma}
\label{subsubsec:daeho}
Recall from the warm-up that the sign estimates $\bsigma_1, \bsigma_2, \ldots, \bsigma_n$ satisfy $\ex{}{\bsigma_i} \ge p$, which implies
\[
    \Prx_{\bsigma}\sbra{\frac{1}{n}\sum_{i=1}^{n}\bsigma_i \ge \Omega(p)} \ge \Omega(p)\,.
\]
Now, note that if we were able to replace the R.H.S.~above with $\Omega(1)$ instead of $\Omega(p)$, then this stronger lower bound would imply that $O(1)$ independent sign estimates would suffice for ``catching'' the $\Omega(p)$ bias.  Then, the sample complexity (up to polylogarithmic factors) would be \smash{$n^2p^2 + \frac{n}{p^2}$} (cf.~\Cref{eq:stitch}), which reduces to \smash{$\widetilde O(n^{3/2})$} if we set $p = n^{-1/4}$.

We first note that the property $\ex{}{\bsigma_i} \ge p$ alone does not imply the desired high-probability bound, namely,
\[
	\pr{}{\frac{1}{n}\sum_{i=1}^{n}\bsigma_i \ge \Omega(p)} \ge \Omega(1).
\]
Suppose for the sake of simplicity that $n$ is even, and imagine that $\bsigma \in \{\pm 1\}^n$ is drawn from the following distribution:
\begin{itemize}
	\item With probability $p$, $\bsigma = (+1, +1, \ldots, +1)$.
	\item With the remaining probability $1 - p$, $\bsigma$ is chosen as a uniformly random permutation of $n/2$ copies of $+1$ and $n/2$ copies of $-1$.
\end{itemize}
It can be easily verified that the above distribution satisfies $\ex{}{\bsigma_i} = p$, yet $\frac{1}{n}\sum_{i=1}^{n}\bsigma_i > 0$ only holds with probability $p$.

Therefore, the crux is to show that the $\bsigma_i$'s cannot be too correlated (as in the example above). 
Concretely, proving the following will suffice for our purposes:

\begin{lemma} 
\label{prop:BE-bash}
    Suppose $\calD$ is a unate distribution on $\bn$. 
    Suppose $\x{1}, \dots \x{m} \sim \calD$ and let 
    \[
        \bsigma_i := \sign\pbra{\sum_{j=1}^m \x{j}_i}\,.
    \]
    For any $1 \leq i < j \leq n$, it holds that 
    \[
        |\Cov[\bsigma_i, \bsigma_j]| = O\pbra{\frac{1}{\sqrt{m}}}\,.
    \]
\end{lemma} 

% \textcolor{gray!40!white}{
% We note that the proof of \Cref{prop:BE-bash} must use the fact that $\calD$ is unate; without unateness (i.e., only assuming $\ex{x \sim \calD}{x_i} = 1/n$), there is a simple counterexample. Consider the distribution $\calD$ defined as
% \[
% 	\calD(+1, +1, \ldots, +1) = \frac{1 + 1/n}{2},
% \quad
% 	\calD(-1, -1, \ldots, -1) = \frac{1 - 1/n}{2}.
% \]
% Such a distribution satisfies that $\ex{x \sim \calD}{x_i} = 1/n$ for every $i \in [n]$, and the coordinates of $x \sim \calD$ are maximally correlated. Then we always get $\sigma_1 = \sigma_2 = \cdots = \sigma_n$, so the correlations are as high as $\Omega(1)$.}
It is readily verified that~\Cref{prop:BE-bash} fails for non-unate distributions: consider, for example, the two-point distribution $\calD(+1^n) = \calD(-1^n) = 0.5$. 
We defer the proof of~\Cref{prop:BE-bash} to \Cref{subsubsec:daeho-berry-esseen}, and first show why it implies the correctness of~\Cref{alg:unif-unate}.

\subsubsection{Proof of~\Cref{thm:our-unif-test}}
\label{subsec:unif-proof}

\newcommand{\y}[1]{\by^{(#1)}}

We will require the following lemma due to Qiao and Valiant~\cite{qiao2018learning}:

\begin{lemma}[Lemma~B.2 of~\cite{qiao2018learning}]
\label{lemma:QV18}
    Suppose $k \in \N$ and $\delta \in (0, 1/15\sqrt{k})$. 
    Let $P$ and $Q$ be two distributions on the same support with $\dtv(P, Q) \geq \delta$. 
    Then 
    \[
        \dtv(P^{\otimes k}, Q^{\otimes k}) \geq \frac{\delta\sqrt{k}}{15}\,.
    \]
\end{lemma}

\begin{algorithm}[t]

	\
	
	\textbf{Input:}~Access to i.i.d.~samples from unate $\calD$, distance parameter $\eps \in (0,1)$\\[0.25em]
	
	\textbf{Output:}~``Accept'' or ``reject'' 
	
	\ 
	
	$\UnifUnate(\calD, \eps)$:\\[0.5em]
	\begin{enumerate}
			\item Repeat the following $O(1)$ times:
			\begin{enumerate}
				\item Draw $m_1 := \Theta\pbra{\frac{n^{3/2}}{\eps^2}}$ samples $\x{1},\dots,\x{m_1} \sim\calD$. 
				\item Compute $\bsigma\in\bn$ where for $i \in [n]$ we have 
				\[
				    \bsigma_i := \sign\pbra{\sum_{j=1}^{m_1} \x{j}_i} \in \bits\,. 
				\]
				\item Draw $m_2 := \Theta\pbra{\frac{n^{3/2}}{\eps^2}\log\pbra{\frac{n}{\eps}}}$ samples $\y{1},\dots,\y{m_2} \sim\calD$.
                \item For each $j \in [m_{2}]$, let $\bw_j \coloneqq \sum_{i=1}^{n}\bsigma_i \cdot \y{j}_i$ denote the $\bsigma$-weighted Hamming \newline weight of $\y{j}$.
                \item If $|\bw_j| \geq \Theta(\sqrt{n\log (n/\eps)})$ holds for any $j \in [m_2]$, halt and output ``reject.'' 
                \item If $\frac{1}{m_2}\sum_{j=1}^{m_2}\bw_j \ge \Theta(\eps / n^{1/4})$, then halt and output ``reject.'' 
			\end{enumerate}
			\item If the algorithm has not rejected yet, output ``accept.'' 
	\end{enumerate}
	
	\caption{The $\UnifUnate$ algorithm.}
	\label{alg:unif-unate}	
\end{algorithm}

\begin{proofof}{\Cref{thm:our-unif-test}}
    Note that the sample complexity and runtime are evident from~\Cref{alg:unif-unate}. 
    Next, note that if $\calD = \calU$, then $\calD^\sigma = \calU$ for every $\sigma \in \bn$. 
    It now follows from Theorem~4 of \cite{RubinfeldS09} that \UnifUnate{} will output ``accept'' with probability at least $9/10$. 
    In particular, note that in this case, Steps~1(c) through 1(f) of~\UnifUnate{} are identical to the \texttt{TestUniform} algorithm of \cite{RubinfeldS09} run on $\calU$. 
	
    We now show how~\Cref{prop:BE-bash} implies the soundness of~$\UnifUnate$.
    Suppose that the unate distribution $\calD$ is $\eps$-far from uniform.
    As before, we write  
    \[
        \mu_i \coloneqq \Ex_{\bx\sim\calD}\sbra{\bx_i}\,.
    \]
    Since $\calD$ is unate, it follows that there exists some $\sigma^\ast \in \bits^n$ such that $\calD^{\sigma^\ast}$ is monotone. 
    In particular, we can take $\sigma^\ast_i = \sign(\mu_i)$, and so we have $\sigma_i^\ast \mu_i = |\mu_i| \geq 0$.
    Since the TV distance is preserved under bijections, $\calD^{\sigma^\ast}$ is also $\eps$-far from uniform. Applying~\Cref{lemma:RS} to $\calD^{\sigma^\ast}$ gives 
    \[
        \|\mu\|_1
    =   \sumi \sigma^\ast_i \mu_i
    =   \ex{\bx \sim \calD^{\sigma^\ast}}{\sum_{i=1}^{n}\bx_i}
    \geq \eps\,.
    \]
    
    Recall that we use a sample of size $m_1 = \Theta(n^{3/2}/\eps^2)$ in Step 1(a)~and~1(b) to compute $\bsigma_1, \ldots, \bsigma_n \in \{\pm 1\}$. Set a threshold $\theta = \frac{\eps}{100n}$. For every $i \in [n]$ that satisfies $|\mu_i| \geq \theta$, we will now show that  
    \begin{equation} \label{eq:neyman-pearson}
        \ex{}{\sigma^\ast_i\bsigma_i} 
        % \ge \Omega(\sqrt{m_1} \cdot |\mu_i|) 
        \ge \Omega\pbra{\frac{\eps}{n}\cdot \frac{n^{3/4}}{\eps}} = \Omega(n^{-1/4})\,,
    \end{equation}
    using  \Cref{lemma:QV18}. 
    In more detail, let $P$ be the distribution of $\sigma_i^\ast \bx_i$ for $\bx \sim \calD$ and $Q$ be $-P$, i.e. $\by\sim Q$ is obtained by drawing $\by'\sim P$ and setting $\by = -\by'$.
    Note that 
    \[
        \Prx_{\by \sim \calP}[\by = +1] = \frac{1+|\mu_i|}{2}
        \qquad\text{and}\qquad 
        \Prx_{\by \sim \calP}[\by = -1] = \frac{1-|\mu_i|}{2}\,, 
    \]
    with swapped probabilities for $Q$. 
    It follows that $\dtv(P,Q) = \abs{\mu_i} \geq \frac{\eps}{100n}$; note also that 
    \[
        \frac{\eps}{100n} \leq \frac{1}{15\sqrt{m_1}}
    \]
    for appropriate choice of hidden constant in $m_1$. 
    In particular, applying~\Cref{lemma:QV18} gives 
    \[
        \dtv(P^{\otimes m_1}, Q^{\otimes m_1}) \geq \Omega\pbra{\frac{\eps}{n}\cdot \frac{n^{3/4}}{\eps}} = \Omega(n^{-1/4})\,.
    \]
    In order to establish~\Cref{eq:neyman-pearson}, it suffices to show that 
    $\dtv(P^{\otimes m_1}, Q^{\otimes m_1}) = \Theta(1)\cdot\Ex[\sigma_i^\ast\bsigma_i]$. 
    Since $P$ and $Q$ are Bernoulli random variables, their likelihood ratio is monotone in the number of successes (i.e., $+1$ draws), and so by the Neyman--Pearson lemma, the TV-distance maximizing set is $\{\sum_\ell \y{\ell} \geq 0\}$. 
    We may assume that $m_1$ is odd (to avoid handling ties), and so we get
    \begin{align}
        \dtv(P^{\otimes m_1}, Q^{\otimes m_1}) 
        &= P^{\otimes m_1}\pbra{\sum_{\ell=1}^{m_1} \y{\ell} \geq 0} - Q^{\otimes m_1}\pbra{\sum_{\ell=1}^{m_1} \y{\ell} \geq 0} \nonumber \\
        % &= \frac{\Ex_{P}[\bsigma_i] + 1}{2} - \frac{1}{2} 
        &=P^{\otimes m_1}\pbra{\sum_{\ell=1}^{m_1} \y{\ell} \geq 0} - P^{\otimes m_1}\pbra{\sum_{\ell=1}^{m_1} \y{\ell} < 0} \label{eq:dunkin}\\
        &= \Ex[\sigma^\ast_i\bsigma_i]\,, \nonumber
    \end{align}
    from which~\Cref{eq:neyman-pearson} follows readily. 
    Note that \Cref{eq:dunkin} relies on the fact that $P = -Q$. 
    Finally, note also that the same argument shows that $\ex{}{\sigma^\ast_i\bsigma_i} \ge 0$ for every $i \in [n]$, even if $|\mu_i| < \theta$. 
	
	Let $\bX \coloneqq \sum_{i=1}^{n}\mu_i \cdot \bsigma_i = \sum_{i=1}^{n}|\mu_i| \cdot (\sigma^\ast_i\bsigma_i)$. We have
    \[
	\ex{}{\bX}
    \ge	\sum_{i \in [n]: |\mu_i| \ge \theta}|\mu_i| \cdot \ex{}{\sigma^\ast_i\bsigma_i}
    \ge \Omega(n^{-1/4}) \cdot \sum_{i \in [n]: |\mu_i| \ge \theta}|\mu_i|
    =   \Omega(n^{-1/4} \cdot \|\mu\|_1).
    \]
    The first step above holds since both $|\mu_i|$ and $\ex{}{\sigma^\ast_i\bsigma_i}$ are non-negative for every $i \in [n]$. The second step applies $|\mu_i| \ge \theta \implies \ex{}{\sigma^\ast_i\bsigma_i} \ge \Omega(n^{-1/4})$. The third step holds since our choice of $\theta = \frac{\eps}{100n}$ and the fact $\|\mu\|_1 \ge \eps$ together imply
    \[
        \sum_{i \in [n]: |\mu_i| \ge \theta}|\mu_i|
    \ge	\sum_{i=1}^{n}|\mu_i| - n \cdot \theta
    =  \|\mu\|_1 - \frac{\eps}{100}
    \ge \Omega(\|\mu\|_1).
    \]
	
    It remains to show that
    \[
        \Var[\bX] \le O(n^{-1/2} \cdot \|\mu\|_1^2).
    \]
    By Chebyshev's inequality, the variance bound above would imply that, with probability $\Omega(1)$,
    \[
        \bX
    \ge \ex{}{\bX} - O(\sqrt{\Var[\bX]})
    \ge \Omega(n^{-1/4} \cdot \|\mu\|_1)
    \ge \Omega(\eps / n^{1/4}),
    \]
    where the second step applies the following observations:
    \begin{itemize}
        \item $\ex{}{\bX} \ge \Omega(n^{-1/4} \cdot \|\mu\|_1)$, where the hidden constant in $\Omega(\cdot)$ is lower bounded by a universal constant when $m_1 = \Theta(n^{3/2}/\eps^2)$ is sufficiently large.
        \item We will show that $\Var[\bX] \le O(n^{-1/2} \cdot \|\mu\|_1^2)$, where the hidden constant in $O(\cdot)$ goes to zero as the hidden constant in $m_1 = \Theta(n^{3/2}/\eps^2)$ increases.
        \item Therefore, for some careful choice of $m_1$ in the algorithm, the difference $\ex{}{\bX} - O(\sqrt{\Var[\bX]})$ is positive and on the order of $\Omega(n^{-1/4} \cdot \|\mu\|_1)$.
    \end{itemize}
    By a Chernoff bound, we can then catch this $\Omega(\eps / n^{1/4})$ bias using $\widetilde O\left(\frac{n}{(\eps / n^{1/4})^2}\right) = \widetilde O(n^{3/2} / \eps^2)$ additional samples in Step~1(f) of the algorithm.%\red{TODO: This is just Chebyshev's inequality, say this.} 
	
	By \Cref{prop:BE-bash}, we have
	\[
		\Var[\bX]
	=	\sum_{i=1}^{n}\mu_i^2\Var[\bsigma_i] + 2\sum_{1 \le i < j \le n}\mu_i\mu_j\Cov[\bsigma_i, \bsigma_j]
	\le \sum_{i=1}^{n}\mu_i^2\Var[\bsigma_i] + O(n^{-1/2})\cdot \|\mu\|_1^2.
	\]
	The second term above is exactly the desired upper bound. For the first term, we note that
	\[
		\Var[\bsigma_i] = e^{-\Omega(m_1\mu_i^2)},
	\]
	so each term $\mu_i^2\Var[\sigma_i]$ is at most $O(1/m_1)$. Recalling $m_1 = \Theta(n^{3/2} / \eps^2)$ and $\|\mu\|_1 \ge \eps$, we conclude that
    \[
        \sum_{i=1}^{n}\mu_i^2\Var[\bsigma_i]
    \le O(n/m_1)
    =   O(\eps^2 / \sqrt{n})
    \le O(n^{-1/2} \cdot \|\mu\|_1^2),
    \]
    which proves the desired variance bound $\Var[\bX] \le O(n^{-1/2} \cdot \|\mu\|_1^2)$ and completes the proof.
\end{proofof}

\subsubsection{Proof of~\Cref{prop:BE-bash}}
\label{subsubsec:daeho-berry-esseen}

Without loss of generality, we prove the lemma when the unate distribution $\calD$ is monotone; the more general statement then follows from the observation that flipping the $i^\text{th}$ coordinate of $\calD$ leads to a flip in the distribution of $\bsigma_i$, which preserves the magnitude of the covariances.

For $i \in [n]$, we define 
\[
    \bS_i := \sum_{\ell=1}^m \x{\ell}_i 
    \qquad\text{and}\qquad 
    \mu_i := \Ex_{\bx\sim\calD}\sbra{\bx_i}\,. 
\]
In particular, we have $\bsigma_i = \sign(\bS_i) = 2\cdot \mathbf{1}\cbra{\bS_i\geq 0} - 1$, and so 
\begin{equation} \label{eq:cov-in-terms-of-big-sig}
    \Cov[\bsigma_i,\bsigma_j] = 4\pbra{\Pr[\bS_i \geq 0,\bS_j \geq 0] - \Pr[\bS_i \geq 0]\Pr[\bS_j \geq 0] }\,.
\end{equation}

We will prove~\Cref{prop:BE-bash} with $i = 1$ and $j = 2$; that is, we will show that 
\[
    \Cov[\bsigma_1, \bsigma_2] = O\pbra{\frac{1}{\sqrt{m}}}\,.
\]
Define 
\begin{align*}
    p_{00} := \Prx_{\bx\sim\calD}\sbra{(\bx_1,\bx_2) = (-1,-1)}\,,
    \qquad&\qquad
    p_{10} := \Prx_{\bx\sim\calD}\sbra{(\bx_1,\bx_2) = (+1,-1)}\,, \\ 
    p_{01} := \Prx_{\bx\sim\calD}\sbra{(\bx_1,\bx_2) = (-1,+1)}\,,
    \qquad&\qquad
    p_{11} := \Prx_{\bx\sim\calD}\sbra{(\bx_1,\bx_2) = (+1,+1)}\,.
\end{align*}
Thanks to monotonicity of $\calD$, we have $p_{00} \leq p_{10} \leq p_{11}$ and $p_{00}\leq p_{01} \leq p_{11}$. 
It also follows that $\mu_1 = 1 - 2(p_{00} + p_{01}) \geq 0$ and $\mu_2 = 1 - 2(p_{00} + p_{10}) \geq 0$. 
The proof considers two cases depending on the value of $\mu_1 + \mu_2$: 

\subsubsection*{Case 1: $\mu_1 + \mu_2 > 1/2$} 

We record the following easy bound on $\Cov[\bsigma_1, \bsigma_2]$: 
\begin{equation}
\label{eq:easy-cov-ub}
    \abs{\Cov[\bsigma_1,\bsigma_2]} \leq 4\min\cbra{\Pr\sbra{\bS_1 < 0}, \Pr\sbra{\bS_2 < 0}}\,.
\end{equation}
Assuming~\Cref{eq:easy-cov-ub}, a direct application of the Hoeffding bound gives 
\[
    \Pr[\bS_i <0] \leq e^{-m\mu_i^2/2}\,,
    \qquad\text{and so}\qquad
    |\Cov(s_1, s_2)|\leq 4e^{-m/32}\ll O\pbra{\frac{1}{\sqrt{m}}}\,,
\]
where we make use of the fact that $\min\{\mu_1, \mu_2\} \geq 1/4$ by the assumption $\mu_1 + \mu_2 > 1/2$. 

We now justify~\Cref{eq:easy-cov-ub}. 
It will be convenient to write 
\[
    A := \Pr[\bS_1\geq 0, \bS_2\geq0]\,,
    \qquad B := \Pr[\bS_1\geq0]\,, 
    \qquad\text{and}\qquad 
    C=\Pr[\bS_2\geq0]\,.
\]
Note that $\E[\bsigma_1] = 2B - 1$ and $\E[\bsigma_2] = 2C-1$. Writing $D := 1 + A - B - C$, it is readily checked that $D = \Pr[\bS_1 < 0, \bS_2 < 0]$. 
It follows from~\Cref{eq:cov-in-terms-of-big-sig} that  $\Cov[\bsigma_1,\bsigma_2] = 4(A-BC)$. 
In particular, we have 
\[
    \frac{1}{4}\abs{\Cov[\bsigma_1,\bsigma_2]} = \abs{A - BC}\,.
\]
Now, 
\begin{itemize}
    \item If $A \geq BC$, then $A-BC \leq C(1 - B) \leq 1-B$.
    \item If $A < BC$, then $BC - A \leq C-(B+C-1) = 1-B$, since $D \geq 0$ gives $A \geq B+C -1$. 
\end{itemize}
It follows that $|\Cov[\bsigma_1,\bsigma_2]| \leq 4(1-B) = 4\Pr[\bS_1 < 0]$. 
By symmetry, the same holds with $B$ and $C$ swapped, which yields~\Cref{eq:easy-cov-ub}. 

%%%%%%%

\subsubsection*{Case 2: $\mu_1 + \mu_2 \leq 1/2$}

In this case, we will control $\Cov(\bsigma_1,\bsigma_2)$ via a Gaussian approximation to the sums $\bS_1$ and $\bS_2$. 
First, let $\kappa := \Cov[\bx_1,\bx_2]$ for $\bx\sim\calD$. 
Note that $\kappa = p_{00} + p_{11} - p_{01} - p_{10} - \mu_1\mu_2$. 

We will first control the covariance of the Gaussian approximation. 
Consider the standardized random vector
\[
    \bT = (\bT_1, \bT_2) := \frac{1}{\sqrt{m}}(\bS_1 - m\mu_1, \bS_2 - m\mu_2)\,
    \qquad\text{and let}\qquad 
    \Sigma:=\begin{pmatrix} \lambda_1^2 & \kappa \\[2pt] \kappa & \lambda_2^2\end{pmatrix}
\]
where $\lambda_i^2 = 1 - \mu_i^2$. 
We will sometimes write $\mu = (\mu_1, \mu_2)$. 
Finally, let $\rho$ be the correlation between $\bT_1$ and $\bT_2$, that is, 
\[
    \rho := \frac{\kappa}{\lambda_1\lambda_2}\,.
\]

Note that by the central limit theorem, $\bT \to N(0, \Sigma)$ (in distribution) as $m \to \infty$. 
Following~\Cref{eq:cov-in-terms-of-big-sig}, the Gaussian approximation of the correlation is 
\[
    \Cov_{\approx} 
    := 
    4\Bigl(\Pr\big[\bZ_1\geq -\sqrt{m}\mu_1,\bZ_2\geq-\sqrt{m}\mu_1\big]-\Pr\big[\bZ_1\geq-\sqrt{m}\mu_1\big]\Pr\big[\bZ_2\geq-\sqrt{m}\mu_2\big]\Bigr).
\] 
Note that $\Cov_{\approx} = \Cov[\btau_1,\btau_2]$ where $\btau_i = \sign(\bZ_i - \sqrt{m}\mu_i)$. 
Next, we will show that $\Cov_{\approx} = O(m^{-1/2})$ via known estimates for computing bi-variate normal probabilities. 
In particular, applying Equation~3.6 of~\cite{owen1956tables} gives 
\begin{align}
    |\Cov_{\approx}| &= 4\left|\frac{1}{2\pi}\int_0^\rho \frac{1}{\sqrt{1-z^2}}\exp\Bigl(-\frac{1}{2}\frac{(\sqrt{m}\tfrac{\mu_1}{\lambda_1})^2+(\sqrt{m}\tfrac{\mu_2}{\lambda_2})^2-2(\sqrt{m}\tfrac{\mu_1}{\lambda_1})(\sqrt{m}\tfrac{\mu_2}{\lambda_2})z}{1-z^2}\Bigr) \, dz\right|  \nonumber \\
    &\leq \frac{2}{\pi}\int_0^{|\rho|} \frac{1}{\sqrt{1-z^2}}\exp\Bigl(-\frac{m}{2}\left(\frac{(\tfrac{\mu_1}{\lambda_1}-\tfrac{\mu_2}{\lambda_2})^2}{1-z^2}+\frac{2\tfrac{\mu_1}{\lambda_1}\tfrac{\mu_2}{\lambda_2}}{1+z}\right)\Bigr) \, dz \label{abs} \\
    &\leq \frac{2}{\pi}\int_0^{|\rho|} \frac{1}{\sqrt{1-z^2}}\exp\Bigl(-\frac{m}{8}(\tfrac{\mu_1}{\lambda_1}+\tfrac{\mu_2}{\lambda_2})^2\Bigr) \, dz \label{alg} \\
    &\leq \frac{2}{\pi}\sin^{-1}|\rho|\exp\Bigl(-\frac{m}{8}(\mu_1+\mu_2)^2\Bigr) \nonumber \\
    &\leq |\rho|\exp\Bigl(-\frac{m}{8}(\mu_1+\mu_2)^2\Bigr)\,, \label{trig bound}
\end{align} 
where \Cref{abs} relies on the fact that the integrand is larger on the positive $z$ side, \Cref{alg} uses 
\[
    \frac{(\tfrac{\mu_1}{\lambda_1}-\tfrac{\mu_2}{\lambda_2})^2}{1-z^2}+\frac{2\tfrac{\mu_1}{\lambda_1}\tfrac{\mu_2}{\lambda_2}}{1+z} \geq (\tfrac{\mu_1}{\lambda_1}-\tfrac{\mu_2}{\lambda_2})^2+\tfrac{\mu_1}{\lambda_1}\tfrac{\mu_2}{\lambda_2} = \frac{1}{4}(\tfrac{\mu_1}{\lambda_1}+\tfrac{\mu_2}{\lambda_2})^2+\frac{3}{4}(\tfrac{\mu_1}{\lambda_1}-\tfrac{\mu_2}{\lambda_2})^2\,,
\]
and \Cref{trig bound} relies on $\lambda_i\leq 1$ and $\sin^{-1}(\theta) \leq \frac{\pi\theta}{2}$ for $\theta\geq 0$.

We record the following lemma: 

\begin{lemma}\label{absrho}
    For monotone distributions, if $\mu_1+\mu_2 \leq \frac{1}{2}$, then $|\rho|\leq \mu_1+\mu_2$. 
\end{lemma}
\begin{proof}
    Parametrize the probabilities of the monotone distribution as 
    \[p_{00} = a, \quad p_{01}=a+b, \quad p_{10}=a+c, \quad p_{11}=a+d.\]
    Then $\mu_1 = d+c-b$, $\mu_2 = d+b-c$, so $\mu_1+\mu_2 = 2d$ implies $d\leq \tfrac14$. Now 
    \begin{align}
        |\rho| &= \frac{|(d-b-c)-(d+c-b)(d+b-c)|}{\sqrt{(1-\mu_1^2)(1-\mu_2^2)}} \nonumber\\
        &\leq \frac{2}{\sqrt{3}}\left|(d-b-c)+((b-c)^2-d^2)\right| \label{mu} \\
        &\leq \frac{2}{\sqrt{3}}(d+d^2) \label{bc} \\
        &=2d\cdot \frac{1+d}{\sqrt{3}} \leq 2d = \mu_1+\mu_2  \nonumber 
    \end{align} 
    where in \Cref{mu} we use that $(1-\mu_1^2)(1-\mu_2^2) \geq \frac{3}{4}$, and in \Cref{bc} we note that $|d-b-c|\leq d$ and $|b-c|\leq d$. 
\end{proof}
Combining~\Cref{trig bound} with~\Cref{absrho}, we get
\begin{equation} \label{eq:senior-fitness}
    |\Cov_{\approx}| \leq (\mu_1+\mu_2) \exp\Bigl(-\frac{m}{8}(\mu_1+\mu_2)^2\Bigr) = O\pbra{\frac{1}{\sqrt{m}}}\,.
\end{equation}

Next, we will control the error in the Gaussian approximation to $\Cov[\bsigma_1,\bsigma_2]$ itself by applying the multivariate Berry-Eseeen central limit theorem (cf.~\Cref{thm:BE}). 
It is readily verified using \Cref{eq:cov-in-terms-of-big-sig} and the triangle inequality that 
\[
    \frac{1}{4}\pbra{\Cov[\bsigma_1,\bsigma_2] - \Cov_{\approx}} \leq \sum_{i=1}^3 |\alpha_i - \beta_i|
\]
where 
\begin{gather*}
    \alpha_1 = \Prx[\bS_1 \geq 0, \bS_2 \geq 0]\,,\quad  
    \alpha_2 = \Pr[\bS_1 \geq 0]\,,\quad 
    \alpha_3 = \Pr[\bS_2 \geq 0]\,,\\[0.5em] 
    \beta_1 = \Pr\sbra{\bZ_1 \geq -\sqrt{m}\mu_1, \bZ_2 \geq -\sqrt{m}\mu_2}\,,\quad 
    \beta_2 = \Pr\sbra{\bZ_1 \geq -\sqrt{m}\mu_1}\,, \quad 
    \beta_3 = \Pr\sbra{\bZ_2 \geq -\sqrt{m}\mu_2}\,. 
\end{gather*}
\Cref{thm:BE} combined with the fact that both $\mu_1, \mu_2 \leq 1/2$ immediately gives $|\alpha_2-\beta_2|$, $|\alpha_3-\beta_3| = O(m^{-1/2})$. 
In more detail, \Cref{thm:BE} gives
\[
    |\alpha_2-\beta_2| 
    = O\left(\frac{1}{\sqrt{m}}\cdot\frac{[|\bS_1-\mu_1|^3]}{\Ex[(\bS_1-\mu_1)^2]^{3/2}}\right) 
    = O\left(\frac{1}{\sqrt{m}}\cdot\frac{(1+\mu_1)^3]}{(1-\mu_1^2)^{3/2}}\right) \leq O\pbra{\frac{1}{\sqrt{m}}}\,.
\] 
An identical argument gives the same bound for $|\alpha_3-\beta_3|$. 
Finally, in order to show that $|\alpha_1-\beta_1| = O(m^{-1/2})$ using~\Cref{thm:BE}, it suffices to establish that $\E[\|\Sigma^{-1/2}(\bS-\mu)\|_2^3] = O(1)$, where $\bS = (\bS_1, \bS_2)$. 
To see this, first note that 
\[
    \Sigma^{-1/2} = \frac{1}{2}\begin{bmatrix}
        \alpha & \beta \\
        \beta & \alpha
    \end{bmatrix}\begin{bmatrix}
        \lambda_1 & 0 \\ 0 & \lambda_2
    \end{bmatrix}^{-1}
\]
where $\alpha = \frac{1}{\sqrt{1+\rho}}+\frac{1}{\sqrt{1-\rho}}$ and $\beta = \frac{1}{\sqrt{1+\rho}}-\frac{1}{\sqrt{1-\rho}}$. 
Now, 
\begin{align}
    \|\Sigma^{-1/2}(\bS-\mu)\|_2^2 &= \frac{1}{4}\left((\alpha\lambda_1^{-1}(\bS_1-\mu_1)+\beta\lambda_2^{-1}(\bS_2-\mu_2))^2+(\beta\lambda_1^{-1}(\bS_1-\mu_1)+\alpha\lambda_2^{-1}(\bS_2-\mu_2))^2\right) \nonumber\\
    &= \frac{1}{1-\rho^2}\left(\lambda_1^{-2}(\bS_1-\mu_1)^2+\lambda_2^{-2}(\bS_2-\mu_2)^2-\rho\lambda_1^{-1}\lambda_2^{-1}(\bS_1-\mu_1)(\bS_2-\mu_2)\right) \label{alphabeta} \\
    &\leq \frac{4}{1-\rho^2}\left(\lambda_1^{-2}+\lambda_2^{-2}+\rho\lambda_1^{-1}\lambda_2^{-1}\right) \label{variab} \\
    &= 4\cdot\frac{2-(\mu_1^2+\mu_2^2)+\kappa}{(1-\mu_1^2)(1-\mu_2^2)-\kappa^2} \nonumber \\
    &\leq 4\cdot \frac{2-\frac{1}{2}+\frac{1}{4}}{\frac{3}{4}-\left(\frac{4}{9}\right)^2} < 16 \label{kap}, 
\end{align} 
where in \Cref{alphabeta} we note that $\alpha^2+\beta^2 = \frac{4}{1-\rho^2}$ and $2\alpha\beta = -\frac{4\rho}{1-\rho^2}$, in \Cref{variab} we use $|X_i-\mu_i|\leq 2$, and in \Cref{kap} we recall $(1-\mu_1^2)(1-\mu_2^2) \geq \frac{3}{4}$ and $-\tfrac{4}{9} \leq \kappa \leq \tfrac{1}{4}$. 
The result follows immediately. 
\qed

\subsection{Lower Bound}
\label{subsec:unif-lb}

Turning to lower bounds for uniformity testing of unate distribution, we will establish the following: 

\unifLB*

In particular, \Cref{thm:our-unif-lb} implies that the algorithm from~\Cref{subsec:unif-ub} has essentially optimal sample complexity.  
Our proof of~\Cref{thm:our-unif-lb} will rely on a slight modification of the construction used by Rubinfeld and Servedio~\cite{RubinfeldS09} to show a tight lower bound for uniformity testing of monotone distributions. 

\subsubsection{Useful Preliminaries}
\label{subsec:unif-lb-prelims}

Throughout this section, we set a parameter 
\[
    \tau := \Theta\pbra{\frac{\log n}{n}}\,.
\]
We write $\calD_\tau$ for the product distribution over $\bn$ where each marginal is a $\tau$-biased bit. 
In other words, for $x\in\bn$, we have 
\[
    \calD_\tau(x) = \pbra{\frac{1}{2} + \tau}^{\abs{\{i:x_i = +1\}}}\pbra{\frac{1}{2} - \tau}^{\abs{\{i:x_i = -1\}}}\,.
\]
Claim~11 of~\cite{RubinfeldS09} gives the following representation of $\calD_\tau$ which will be useful for our purposes: 
\begin{equation} \label{eq:RS-subcube-decomp}
    \calD_\tau(x) 
    = \sum_{i=0}^n \sum_{\substack{y \in \bn \\ |i : y_i = +1| = i}} (2\tau)^i(1-2\tau)^{n-i} \calU_y(x)
    =: \sum_{y \in \bn} p_y \calU_y(x)\,,
\end{equation} 
where we write $\calU_y$ for the uniform distribution on the monotone subcube rooted at $y$, i.e., the uniform distribution on $\{z \in \bn : z_i \geq y_i~\text{for all}~i\in[n] \}$. 

\begin{proposition}[Claim~7 of~\cite{RubinfeldS09}]
\label{prop:RS-claim-7}
    Let $T \in \N$ and suppose $m \leq 0.1\sqrt{T}$. 
    Let $\bS_1$ be the random variable where a draw is obtained by making $m$ independent draws from $\calD_\tau$. 
    We define a new random variable $\bS_2$ on the same domain as $\bS_1$, where a draw is obtained as follows: 
    % Consider the following random process: 
    \begin{enumerate}
        \item[(i)] First, for $\ell \in [T]$, independently draw $\by^{(\ell)}\sim\bn$ with probability $p_y$ (cf.~\Cref{eq:RS-subcube-decomp}).
        \item[(ii)] Define the distribution $\bq$ over $\bn$ as 
        \[
            \bq(x) := \frac{1}{T}\sum_{\ell=1}^T \calU_{\by^{(\ell)}}(x)\,.
        \]
        Then, we make $m$ independent draws from $\bq$. 
    \end{enumerate}
    We then have 
    \[
        \dtv(\bS_1, \bS_2) \leq \frac{1}{100}\,. 
    \]
\end{proposition}

Note, in particular, that the value of $T$ in \Cref{prop:RS-claim-7} is unconstrained. 
We will also require the following 
due to Rubinfeld and Servedio, which states that for $T = n^3$, 
with high probability the distribution $\bq$ constructed
in~\Cref{prop:RS-claim-7} is far from uniform:

\begin{proposition}[Lemma~13 of~\cite{RubinfeldS09}]
\label{prop:RS-lemma-13}
    Suppose $T = n^3$, and let $\bq$ be as in~\Cref{prop:RS-claim-7}. 
    Then with probability $1-o(1)$ over the random choice of $\bq$, we have 
    \[
        \dtv(\bq, \calU) \geq 1 - \frac{1}{n^{9}}\,,
    \]
    where $\calU$ denotes the uniform distribution over $\bn$. 
\end{proposition}

We note that Rubinfeld and Servedio~\cite{RubinfeldS09} proved their Lemma~13 with $T= n^2$, 
but we require a larger $T$ for our result. 
The variant above readily follows from inspecting their proof and adjusting the constant hidden by $\Theta(\cdot)$ in $\tau$. 

\subsubsection{Proof of~\Cref{thm:our-unif-lb}} 

% \snote{I think we wanted to update the notation in this section, right?}

Let $T = n^3$, $m = \Theta(n^{3/2}/\log^2(n))$, and let $\calD_\tau$ be as in the previous section. 
(Note that $m \leq 0.1\sqrt{T}$ for $n$ large enough.)
We will show that there exists a distribution $\calF$ over unate distributions over $\bn$ such that 
\begin{enumerate}
    \item With probability $1-o(1)$ over the draw of $\bq'\sim\calF$, $\dtv(\bq', \calU) \geq 1-n^{-9}$. 
    \item The statistical distance between (a) the distribution of $m$ independent samples from $\calU$, and (b) the draw of $m$ independent samples from $\bq'$ where $\bq' \sim\calF$, is at most $0.2$.  
\end{enumerate}
Items~1 and~2 immediately imply that any algorithm with performance guarantee as in~\Cref{thm:our-unif-test} must draw $m = \Omega(n^{3/2})$ samples.

%Suppose, for the sake of notational simplicity, that $n$ is even.  
We now define the family $\calF$ of hard distributions. 
A distribution $\bq'\sim\calF$ is generated as follows: 
\begin{enumerate}
    \item[(a)] First generating $\bq$ as in \Cref{prop:RS-claim-7} and drawing $\bsigma\sim\bn$.
    \item[(b)] Returning the distribution $\bq' := \bq^{\bsigma}$ (cf.~\Cref{notation:D-sigma}). 
\end{enumerate}
Note that since every $\bq$ generated as in~\Cref{prop:RS-claim-7} is monotone, it follows that every $\bq' = \bq^{\bsigma}$ is a unate distribution. 
Item~1 follows immediately from~\Cref{prop:RS-lemma-13}; all that remains is to establish Item~2 above. 

It will be convenient to write $\calD'$ for the distribution of $m$ independent samples from $\bq'$ after drawing $\bq'\sim\calF$. 
For $\bsigma$ as in Item~(a) above, we have 
\begin{align*}
    \dtv\pbra{\calD', \calU^{\otimes m}} 
    &\leq \dtv\pbra{\calD', (\calD_\tau^{\bsigma})^{\otimes m}} + \dtv\pbra{(\calD_\tau^{\bsigma})^{\otimes m}, \calU^{\otimes m}} \tag{Triangle inequality} \\ 
    &= \dtv\pbra{\bS_1, \bS_2} + \dtv\pbra{(\calD_\tau^{\bsigma})^{\otimes m}, \calU^{\otimes m}}  \\ 
    &\leq \frac{1}{100} + \dtv\pbra{(\calD_\tau^{\bsigma})^{\otimes m}, \calU^{\otimes m}}\,,  \tag{\Cref{prop:RS-claim-7}}
\end{align*}
where $\bS_1, \bS_2$ are as in~\Cref{prop:RS-claim-7}. 
Note that in the first item above, $(\calD_\tau^{\bsigma})^{\otimes m}$ refers to the distribution of $m$ samples obtained by first drawing $\bsigma$ and then drawing $m$ independent samples from $\calD_\tau^{\bsigma}$. 

The remainder of the argument will establish
\begin{equation} \label{eq:unif-LB-final-goal}
    \dtv\pbra{(\calD_\tau^{\bsigma})^{\otimes m}, \calU^{\otimes m}} 
    \leq 0.1\,, 
\end{equation}
which will complete the proof of~\Cref{thm:our-unif-lb}. 
Since each marginal is independent in both $\calD_\tau^{\bsigma}$ and $\calU$, it follows from additivity of the KL divergence that 
\begin{align*}
    \dtv\pbra{(\calD_\tau^{\bsigma})^{\otimes m}, \calU^{\otimes m}} 
    &\leq 
    \sqrt{\frac{1}{2}\dkl\pbra{(\calD_\tau^{\bsigma})^{\otimes m}, \calU^{\otimes m}}}  \tag{\Cref{eq:distance-relationships}} \\ 
    &= \sqrt{\frac{1}{2}\sumi \dkl\pbra{(\calD_\tau^{\bsigma})_i^{\otimes m}, (\calU)_i^{\otimes m}}}\,,
\end{align*}
where we use the subscript $i\in [n]$ to denote the marginal of the draw of $m$ samples, each restricted to the $i^\text{th}$ bit. 
Note that for any fixed $i\in [n]$, the quantity $\dkl\pbra{(\calD_\tau^{\bsigma})_i^{\otimes m}, (\calU)_i^{\otimes m}}$ is the KL divergence between (i) the uniform mixture of $\Binomial(m,1/2-\tau)$ and $\Binomial(m, 1/2+\tau)$, and (ii) $\Binomial(m, 1/2)$. 

Let $p_\pm:=\tfrac12\pm\tau$ with $0<\tau<\tfrac12$, and define the distributions 
\[
    P(k)=\tfrac12\binom{m}{k}p_-^{\,k}(1-p_-)^{m-k} +\tfrac12\binom{m}{k}p_+^{\,k}(1-p_+)^{m-k},
    \qquad
    Q(k)=\binom{m}{k}\,2^{-m}
\]
for notational convenience. Note that $P$ is the uniform mixture from above, while $Q$ is $\Binomial(m, 1/2 + \tau)$\,. 
Finally, let 
\[
    f_{\pm}(k) := \frac{P_{\pm}(k)}{Q(k)} = \frac{P_\pm(k)}{Q(k)}
    =\bigl(2p_\pm\bigr)^{k}\bigl(2(1-p_\pm)\bigr)^{m-k}\,,
\]
and note that $\E_{Q}[f_{\pm}] = 1$. 
We have from~\Cref{eq:distance-relationships} that 
\begin{align*}
    \dkl\pbra{(\calD_\tau^{\bsigma})_i^{\otimes m}, (\calU)_i^{\otimes m}}
     &\leq \dchi\pbra{(\calD_\tau^{\bsigma})_i^{\otimes m}, (\calU)_i^{\otimes m}} \\
    &=\Ex_Q\!\left[\left(\tfrac12(f_-+f_+)-1\right)^{\!2}\right]\\
    &=\tfrac14\,\Ex_Q\!\bigl[(f_-+f_+)^2\bigr]-1\\
    &=\tfrac14\Bigl(\Ex_Q[f_-^2]+\Ex_Q[f_+^2]+2\,\Ex_Q[f_-f_+]\Bigr)-1\\
    &=\tfrac12\Bigl((1+4\tau^2)^m+(1-4\tau^2)^m\Bigr)-1\\
    &=8m(m-1)\tau^4+O\!\bigl(m\tau^6\bigr)\,.
\end{align*}
Combining this with the earlier bound, we get 
\[
    \dtv\pbra{(\calD_\tau^{\bsigma})^{\otimes m}, \calU^{\otimes m}} 
    \leq O(1)\cdot \sqrt{n m^2 \tau^4} \ll 1 ~\qquad\text{for}~m = \Theta\pbra{\frac{n^{3/2}}{\log^2 n}}\,,
\]
which establishes~\Cref{eq:unif-LB-final-goal} and completes the proof of~\Cref{thm:our-unif-lb}. \qed

%% file: sections/subcube.tex
\section{Testing Unateness with Subcube Conditioning}
\label{sec:subcube}

We now turn to the proofs of~\Cref{thm:our-subcube-UB,thm:subcube-lower-bound}, starting with the former. 

\subsection{Upper Bound}

We start by recalling \Cref{thm:our-subcube-UB}:

\subcubeUB*

\subsubsection{Preliminaries from~\cite{CCRSW25}}
\label{subset:erik-technical-lemma}

Chakrabarty, Chen, Ristic, Seshadhri, and Waingarten~\cite{CCRSW25} gave an $\widetilde O(n/\eps^2)$-query algorithm for monotonicity testing of a distribution over $\{\pm 1\}^n$ in the subcube conditional model. 
Our $\widetilde O(n^{3/2}/\eps^2)$-query algorithm for unateness testing will rely on their principal technical lemma, which we state after introducing some relevant notation. 

% Throughout, let $\calD$ be a distribution over $\{\pm 1\}^n$ that is $\eps$-far from unate. 
For any $x \in \{\pm 1\}^n$ and $i \in [n]$, we define
\[
    \delta_{x,i} \coloneqq \frac{\calD(x^{i\to -1}) - \calD(x^{i\to +1})}{\calD(x^{i\to -1}) + \calD(x^{i\to +1})} \in [-1, 1]
\]
as the bias on the edge $\{x, x^{\oplus i}\}$. 
Note that if $\calD$ is monotone, we always have $\delta_{x,i} \le 0$, so a positive $\delta_{x,i}$ witnesses a violation of monotonicity.

Central to the analysis in~\cite{CCRSW25} is the following lemma, which is implicit in the proof of~\cite[Lemma 2.3]{CCRSW25}. 
For $x \in \R$, we write $(x)_+ \coloneqq \max\{x, 0\}$.

\begin{lemma}[\cite{CCRSW25}] 
\label{lemma:CCRSW}
    For any $\calD$ that is $\eps$-far from monotone, there exists a positive integer $w \le O(\log(n/\eps))$ such that 
    \[
        \Prx_{\bx \sim \calD \atop \bi \sim [n]}\sbra{(\delta_{\bx, \bi})_+^2 \ge 2^{-w}} \ge \widetilde{\Omega}\left(\frac{2^w \cdot \eps^2}{n}\right).
    \]
\end{lemma}

The proof of \Cref{lemma:CCRSW} relies on a real-valued ``directed'' isoperimetric inequality proved in~\cite{CCRSW25}.  
In particular, it follows from~\Cref{lemma:CCRSW} that if we draw a random $\bw \le O(\log(n/\eps))$ with probability $\bw = w$ proportional to $2^{-w}$, we have
    \begin{equation} \label{eq:MIT-tech-review}
        \Prx_{\bw, \bx \sim \calD, \bi \sim [n]}\sbra{(\delta_{\bx, \bi})_+^2 \ge 2^{-\bw}}
    \ge \sum_{w=1}^{O(\log(n/\eps))}2^{-w}\cdot\Prx_{\bx \sim \calD \atop \bi \sim [n]}\sbra{(\delta_{\bx, \bi})_+^2 \ge 2^{-w}}
    \ge \widetilde\Omega\left(\frac{\eps^2}{n}\right).
    \end{equation}

Give the lemma above, the algorithm of~\cite{CCRSW25} is natural: we simply sample $\widetilde O(n/\eps^2)$ triples $(\bw, \bx, \bi)$, in the hope of finding at least one with $(\delta_{\bx, \bi})_+^2 \ge 2^{-\bw}$. For each triple, we make $\widetilde O(2^w)$ queries to the subcube (or edge) $\{x, x^{\oplus i}\}$ to check whether $(\delta_{x,i})_+^2 \ge 2^{-w
}$ holds. If this holds for any triple, we reject distribution $\calD$. 
Over the randomness in $\bw$, we make $\sum_{w=1}^{O(\log(n/\eps))}2^{-w}\cdot\widetilde O(2^w) = \widetilde O(1)$ queries for each triple, so the overall query complexity is $\widetilde O(n/\eps^2)$.

\subsubsection{Proof of~\Cref{thm:our-subcube-UB}} 

We first introduce some additional notation before proving~\Cref{thm:our-subcube-UB}. 
Recall from \Cref{notation:D-sigma} that, for $\sigma \in \bn$, $\calD^\sigma$ is the distribution of 
\[
    \bx^\sigma \coloneqq (\sigma_1 \bx_1, \sigma_2 \bx_2, \ldots, \sigma_n \bx_n)\,,
\]
where $\bx \sim \calD$, i.e., $\calD^\sigma$ is obtained from $\calD$ by flipping the coordinates indexed by $\{i \in [n]: \sigma_i = -1\}$. 
Recall that, with respect to distribution $\calD$, we defined
\[
    \delta_{x,i} \coloneqq \frac{\calD(x^{i\to -1}) - \calD(x^{i\to +1})}{\calD(x^{i\to -1}) + \calD(x^{i\to +1})}\,.
\]
In the following, we will write $\delta^{\calD}_{x,i}$ to avoid confusion and emphasize dependence on $\calD$. 
While the next identity is notation-heavy, it states an intuitive fact: if $\calD^\sigma$ has a bias on an edge adjacent to $x$ along the $i^\text{th}$ direction, $\calD$ should also have a bias at $x^\sigma$, albeit flipped by $\sigma_i$. 
More formally, 
\begin{align}
    \delta^{\calD^\sigma}_{x,i}
&=  \frac{\calD^\sigma(x^{i\to -1}) - \calD^\sigma(x^{i\to +1})}{\calD^\sigma(x^{i\to -1}) + \calD^\sigma(x^{i\to +1})} \nonumber \\
&=  \frac{\calD((x^{i\to -1})^\sigma) - \calD((x^{i\to +1})^\sigma)}{\calD((x^{i\to -1})^\sigma) + \calD((x^{i\to +1})^\sigma)} \tag{$\calD^\sigma(x) = \calD(x^\sigma)$} \nonumber \\
&=  \frac{\sigma_i \cdot \left[\calD((x^\sigma)^{i\to -1}) - \calD((x^\sigma)^{i\to +1})\right]}{\calD((x^\sigma)^{i\to -1}) + \calD((x^\sigma)^{i\to +1})} \nonumber \\
&=  \sigma_i \cdot \delta^P_{x^\sigma, i}\,. \label{eq:banana} 
\end{align}

We now record a corollary of~\Cref{lemma:CCRSW} which will be crucial to our analysis of~\Cref{alg:subcube-unate}. 
Since $\calD$ is $\eps$-far from unate, for every sign pattern $\sigma \in \{\pm 1\}^n$, the distribution $\calD^\sigma$ is $\eps$-far from monotone. 
Applying~\Cref{lemma:CCRSW} (or more precisely, \Cref{eq:MIT-tech-review}) to $\calD^\sigma$ then gives 
\begin{align}
    \wt{\Omega}\pbra{\frac{\eps^2}{n}} 
    &\leq \Prx_{\bw, \bx\sim \calD^\sigma, \bi \sim [n]}\sbra{(\delta^{\calD^\sigma}_{\bx, \bi})_+^2 \ge 2^{-\bw}} \nonumber \\
    &= \Prx_{\bw, \bx\sim \calD^\sigma, \bi \sim [n]}\sbra{(\sigma_i\cdot\delta^{\calD}_{\bx^\sigma, \bi})_+^2 \ge 2^{-\bw}} \tag{\Cref{eq:banana}} \nonumber \\
    &= \Prx_{\bw, \bx\sim \calD, \bi \sim [n]}\sbra{(\sigma_i\cdot\delta^{\calD}_{\bx, \bi})_+^2 \ge 2^{-\bw}}\,, \label{eq:super-nice-bakery}
\end{align}
where \Cref{eq:super-nice-bakery} relies on the fact that for $\bx\sim\calD^\sigma$, the random variable $\bx^\sigma$ is distributed as $\calD$. 
We record this consequence of~\Cref{lemma:CCRSW}  for future use: 

\begin{corollary} 
\label{cor:CCRSW}
    Suppose $\calD$ is $\eps$-far from unate. 
    Then for every $\sigma\in\bn$, we have 
    \[
         \Prx_{\bw, \bx\sim \calD, \bi \sim [n]}\sbra{(\sigma_i\cdot\delta^{\calD}_{\bx, \bi})_+^2 \ge 2^{-\bw}} 
         \geq 
         \wt{\Omega}\pbra{\frac{\eps^2}{n}}\,.
    \]
\end{corollary}

\Cref{cor:CCRSW} suggests a natural test to detect distributions that are far from unate and forms the basis of~\Cref{alg:subcube-unate}: iterating over each coordinate $i\in [n]$, we reject if there exist two pairs $(\bw, \bx)$ and $(\bw', \bx')$ such that 
\begin{equation} \label{eq:subcube-reject-condition}
    (\delta^\calD_{\bx,i})_{+}^2 \ge 2^{-\bw}
    \quad\text{and}\quad
    (-\delta^\calD_{\bx',i})_{+}^2 \ge 2^{-\bw'}\,.
\end{equation}
Intuitively, the former condition rules out the possibility that $\calD$ is monotone in direction $i$, and the latter rules out that $\calD$ is anti-monotone in direction $i$. 
Note that if $\calD$ is unate, then both conditions above can never hold simultaneously. 

\begin{algorithm}[t]

	\
	
	\textbf{Input:}~Subcube query access to $\calD$, distance parameter $\eps \in (0,1)$\\[0.25em]
	
	\textbf{Output:}~``Accept'' or ``reject'' 
	
	\ 
	
	$\SubcubeUnate(\calD, \eps)$:\\[0.5em]
	\begin{enumerate}
		  \item Set $k := \wt{\Theta}\pbra{\frac{\sqrt{n}}{\varepsilon^2}}$ and repeat Step~2 $\Theta(1)$-times: 
            \item For $i \in [n]$: 
            \begin{enumerate}
                \item For $\ell \in [k]$:
                \begin{enumerate}
                    \item Sample $\bx\sim\calD$ and consider the restriction $\brho \in \{\pm1,\ast\}^n$ given by $\brho_j = \bx_j$ \newline if $j\neq i$ and $\brho_{i} = \ast$. 
                    \item Draw $\bw\in[\Theta(\log(n/\eps))]$ where $\Pr[\bw = t] \propto 2^{-t}$. 
                    \item Query the conditional distribution $\calD\mid_{\brho}$ independently
                    \[
                        T(\bw) := \Theta(1)\cdot2^{\bw}\cdot \ln(12nk/\delta_{\mathrm{cmp}})
                    \]
                    times for $\delta_{\mathrm{cmp}}:=n^{-3}$.
                    Let $N_+^{(\ell)}$ and $N_-^{(\ell)}$ be the numbers of returns with
                    coordinate $i$ equal to $+1$ and $-1$, respectively
                    ($\bN_+^{(\ell)}+\bN_-^{(\ell)}=T(\bw)$). \newline Define the empirical edge bias
                    \[
                    \btau_\ell \;\coloneqq\; \frac{\bN_-^{(\ell)}-\bN_+^{(\ell)}}{\bN_-^{(\ell)}+ \bN_+^{(\ell)}}\,,
                    \qquad
                    \text{and set }\ c_{\bw}\coloneqq \frac{2^{-\bw/2}}{4}\,.
                    \]
                    Mark the trial $\ell$ as a \emph{positive witness} if $\btau_\ell\ge c_{\bw}$,
                    and as a \emph{negative witness} \newline if $\btau_\ell < -c_{\bw}$.
                \end{enumerate}
                \item If there exist $\ell_1, \ell_2 \in [k]$ such that $\ell_1$ is a positive witness and $\ell_2$ is a negative \newline witness, then halt and output ``reject.'' 
            \end{enumerate}
            \item If the algorithm has not rejected, output ``accept.''
	\end{enumerate}
	
	\caption{
            The $\SubcubeUnate$ algorithm. 
        }
	\label{alg:subcube-unate}	
\end{algorithm}

We now turn to the proof of~\Cref{thm:our-subcube-UB}:

\begin{proof}[Proof of~\Cref{thm:our-subcube-UB}]
    For a fixed triple $(i,\bx,\bw)$, the expected number of conditional samples is
    \[
        \sum_{t=1}^{\Theta(\log(n/\varepsilon))}\Pr[\bw=t]\cdot T(t)
        =\sum_t \Theta(2^{-t})\cdot \Theta(2^{t}\log(nk))
        =\widetilde O(1).
    \]
    We perform $\Theta(n\cdot k)=\widetilde \Theta(n^{3/2}/\varepsilon^2)$ triples,
    so the total query complexity is $\widetilde O(n^{3/2}/\varepsilon^2)$. 
    
    By Hoeffding’s inequality,
    \[
        \Pr\!\left[\big|\btau_\ell-\delta^{\calD}_{\bx,i}\big|
        \ \ge\ \frac{1}{8}\,2^{-\bw/2}\right]
        \ \le\ 2\exp\!\Big(-\tfrac12\cdot (\frac{1}{8}\,2^{-\bw/2})^2\,T(\bw)\Big)
        \ \le\ \frac{\delta_{\mathrm{cmp}}}{6nk},
    \]
    for a sufficiently large absolute constant in the definition of $T(\bw)$.
    Taking a union bound over all $nk$ trials, with probability at least
    $1-\delta_{\mathrm{cmp}}/6$ we have, simultaneously for all trials,
    \begin{equation} \label{eq:good-event}
        \btau_\ell \in \Big[\delta^{\calD}_{\bx,i}-\frac{1}{8}\,2^{-\bw/2},\;
        \delta^{\calD}_{\bx,i}+\frac{1}{8}\,2^{-\bw/2}\Big].
    \end{equation}
    For the remainder of the argument, we condition on the event that \Cref{eq:good-event} holds for all trials. 

    \paragraph{Completeness.} Conditioned on the above event, it is readily seen that  \SubcubeUnate{} has one-sided error. 
    In particular, suppose $\calD$ is a unate distribution. 
    It follows that there exists $\sigma\in\bn$ such that $\calD^\sigma$ is monotone. 
    In particular, $\sigma_i \cdot \delta^{\calD}_{x,i} \leq 0$ for all $x\in\bn$ and $i\in[n]$. 
    Fix a coordinate $i \in [n]$ and suppose $\sigma_i = +1$. 
    We have $\delta^{\calD}_{\bx,i}\le 0$ for every $\bx$, and conditioned on~\Cref{eq:good-event},
    \[
        \btau_\ell \le\ \tfrac18\,2^{-\bw/2}\ <\ c_{\bw}= \frac{2^{-\bw/2}}{4}\,,
    \]
    so this trial cannot create a positive witness for coordinate $i$. 
    An identical argument gives that if $\sigma_i = -1$, then the trial cannot create a negative witness for coordinate $i$. 
    It follows that, conditioned on~\Cref{eq:good-event}, the algorithm will output ``accept'' with probability $1$. 

    \paragraph{Soundness.} We now turn to soundness of $\SubcubeUnate(\calD, \eps)$.
    Suppose $\calD$ is $\eps$-far from every unate distribution, we will show that $\SubcubeUnate(\calD,\eps)$ will output ``reject'' with probability at least $9/10$. 
    For $\bw$ distributed as in Step~2(a).ii of~\Cref{alg:subcube-unate} and $i\in[n]$, let 
    \[
        q^{+}_i \coloneqq \Prx_{\bw,\bx \sim \calD}\sbra{(\delta^{\calD}_{\bx, i})_+^2 \ge 2^{-\bw}} 
        \qquad\text{and}\qquad 
        q^{-}_i \coloneqq \Prx_{\bw,\bx \sim \calD}\sbra{(-\delta^{\calD}_{\bx, i})_+^2 \ge 2^{-\bw}} 
    \]
    be the probabilities that, over the randomness in $\bw$ and $\bx$, the pair $(\bw, \bx)$ witnesses that the $i^\text{th}$ coordinate cannot be monotone and anti-monotone respectively. 
    \Cref{cor:CCRSW} directly implies that 
    \begin{equation} \label{eq:corCCRSW}
        m:=\frac{1}{n} \sum_{i=1}^n m_i 
        \geq \widetilde{\Omega}\!\left(\frac{\varepsilon^{2}}{n}\right)\,,
        \qquad\text{where}~m_i \coloneqq \min \left\{ q_i^{+}, q_i^{-} \right\}\,.
    \end{equation}
    Let $r_i$ be the probability that $\SubcubeUnate(\calD,\eps)$ outputs ``reject'' in Step~2(b) on coordinate $i$. 
    Thanks to independence, we have 
    \[
        \Prx\sbra{\SubcubeUnate(\calD,\eps)~\text{outputs ``accept''}} = \prod_{i=1}^n (1-r_i) \leq \exp\pbra{-\sumi r_i}\,,
    \]
    and so the remainder of the argument will establish that $\sumi r_i = \Omega(1).$

    First, note that 
    \begin{equation} \label{eq:quarters}
        r_i 
        \geq \pbra{1-(1-q_i^+)^{k/2}}\pbra{1-(1-q_i^-)^{k/2}}
        \geq \pbra{1-(1-m_i)^{k/2}}^2\,,
    \end{equation}
    where the first inequality holds because the probability of returning ``reject'' after $k$ iterations of Steps~2(a).i--iii is at least the probability of detecting a positive witness among the first $k/2$ iterations and a negative witness among the last $k/2$ iterations. 
    The following lemma will be useful: 

    \begin{lemma}\label{inflection}
        For $K\geq1$, the function $f(x) = (1-(1-x)^K)^2$ is convex on the interval $(0, \ln(2)/K)$. 
    \end{lemma}
    \begin{proof}
    We explicitly compute the first and second derivatives of the function $f$. 
    Note that 
    \[
        f'(x) = 2K\pbra{1-(1-x)^K}(1-x)^{K-1}\,.
    \]
    Taking another derivative gives 
    \begin{align*}
         f''(x) 
         &= 2K^2(1-x)^{2K-2}-2K(K-1)(1-(1-x)^K)(1-x)^{K-2} \\
        &= 2K(1-x)^{K-2}((2K-1)(1-x)^K-(K-1))\,.
    \end{align*}
    So it suffices to show
    \[
        (2K-1)\Bigl(1-\tfrac{\ln 2}{K}\Bigr)^K \geq K-1\,, 
    \]
    or equivalently $F(K)\geq 0$ where $F(x):=x\ln\Bigl(1-\tfrac{\ln 2}{x}\Bigr)-\ln\!\Bigl(\tfrac{x-1}{2x-1}\Bigr)$. Since
    \[
    \lim_{x\to+\infty}F(x)=-\ln 2-\ln \tfrac12=0,
    \]
    showing that $F'(x)<0$ for all $x\in(1, +\infty)$ will complete the proof. 

    Note that 
    \[
        F'(x) = \ln\Bigl(1-\frac{1}{x/\ln2}\Bigr)+\frac{1}{x/\ln2-1}-\frac{1}{(x-1)(2x-1)}.
    \]
    Furthermore, the function $h(y) := \ln\Bigl(1-\frac{1}{y}\Bigr)+\frac{1}{y-1}$ has derivative $h'(y)=-\frac{1}{y(y-1)^2}$, and is thus strictly decreasing on $(1, +\infty)$. If we could show that
    \begin{equation} \label{eq:hooray-daeho}
        h(x) = \ln\Bigl(1-\frac{1}{x}\Bigr)+\frac{1}{x-1}\leq\frac{1}{(x-1)(2x-1)}
    \end{equation}
    holds for all $x \in (1, +\infty)$,
    we would have the desired inequality
    \[
        F'(x)
    =   h(x / \ln 2) - \frac{1}{(x-1)(2x-1)}
    <   h(x) - \frac{1}{(x-1)(2x-1)}
    \le 0.
    \]
    \Cref{eq:hooray-daeho} follows immediately from the well-known/readily verified inequality $2a \leq (2+a)\cdot\ln(1+a)$ for $a\geq0$ by taking $a=(x-1)^{-1}$. 
    \end{proof}
    
    Returning to the soundness of \SubcubeUnate, note that if there exists a coordinate $j$ with $m_j \geq \tfrac{2\ln2}{k}$, then thanks to \Cref{eq:quarters} we have 
    \begin{align*}
       r_j \geq (1-(1-m_j)^{k/2})^2 
        \geq \left(1-\left(1-\frac{\ln2}{k/2}\right)^{k/2}\right)^2
        \geq (1-e^{-\ln2})^2 = \frac{1}{4}.
    \end{align*} 
    Consequently, the amplification in Step~1 ensures that one of the $\Theta(1)$ independent runs will output ``reject'' with probability at least $9/10$. 

    On the other hand, suppose $m_i< \tfrac{2\ln2}{k}$ for all $i \in [n]$. 
    Thanks to \Cref{inflection} and Jensen's inequality, we then have 
    \[
        \sumi r_i 
        \geq \sumi \pbra{1 - (1-m_i)^{k/2}}^2 
        \geq n\pbra{1 - (1-m)^{k/2}}^2\,.    
    \]
    It follows that $\sumi r_i \geq 99/100$ from~\Cref{eq:corCCRSW} and our choice of $k$ by an easy application of Bernoulli's inequality $(1+u)^r \geq 1+ru$ for $r,u \geq 0$. 
\end{proof}

%% file: sections/subcube-lower-bound.tex
\subsection{Lower Bound}
We prove the following lower bound for testing unateness of distributions in the subcube model, which applies to all testers that only make subcube queries on either the entire hypercube (to get an unconditional sample) or a small subcube.

\subcubeLB*

Recall that our tester for \Cref{thm:our-subcube-UB}, like the monotonicity tester of~\cite{CCRSW25}, only makes subcube queries on edges (i.e., one-dimensional subcubes) in addition to drawing unconditional samples. In other words, our tester satisfies the precondition of \Cref{thm:subcube-lower-bound} with $d = 1$. The theorem then implies that all such testers must make at least $\widetilde\Omega(n^{2/3})$ queries in total.

We prove this lower bound by reducing the unateness testing of a Boolean function $f: \bits^n \to \zo$ to that of a distribution over $\bits^n$.\footnote{In this section, we assume that the co-domain of the Boolean function is $\zo$ (instead of $\bits$), which is more convenient when dealing with probability mass functions.} For a non-zero Boolean function $f$, let $\calD_f$ denote the uniform distribution over $f^{-1}(1) \coloneqq \{x \in \bits^n: f(x) = 1\}$. \Cref{thm:subcube-lower-bound} is an immediate consequence of the following observations and a lower bound of Chen, Waingarten, and Xie~\cite{CWX17stoc} for testing unateness of Boolean functions:
\begin{itemize}
    \item (\Cref{lemma:unate-function-to-dist}) If function $f: \bits^n \to \zo$ is unate, distribution $\calD_f$ is also unate. If $f$ is $\eps$-far from unate, $\calD_f$ is $\Omega(\eps)$-far from unate.
    \item (\Cref{lemma:function-to-dist-reduction}) Given membership query access to $f: \bits^n \to \zo$, we can sample from $\calD_f$ and its restrictions to small subcubes. Assuming that $f$ is $\eps$-far from the zero function, $O(1/\eps)$ queries suffice for drawing an unconditional sample from $\calD_f$. Moreover, $2^d$ queries suffice for sampling from $\calD_f$ restricted to a $d$-dimensional subcube.
    \item By~\cite[Theorem 2]{CWX17stoc}, testing the unateness of an $n$-variable Boolean function requires $\widetilde\Omega(n^{2/3})$ queries.
\end{itemize}

We state and prove \Cref{lemma:unate-function-to-dist,lemma:function-to-dist-reduction} in the rest of this section.

\begin{lemma}
\label{lemma:unate-function-to-dist}
    The following two claims hold for every $f: \bits^n \to \zo$ and $\eps \ge 0$:
    \begin{itemize}
        \item If $f$ is unate, $\calD_f$ is unate.
        \item If $f$ is $\eps$-far from unate, $\calD_f$ is $(\eps/2)$-far from unate.
    \end{itemize}
\end{lemma}

\begin{proof}
    Let $g: \bits^n \to [0, 1]$ be the probability mass function (PMF) of $\calD_f$. The first claim easily follows from the observation that $g$ is exactly $\alpha \coloneqq 1 / |f^{-1}(1)|$ times $f$, and is thus unate for every unate function $f$.
    
    For the second claim, we prove its contrapositive: if $\calD_f$ is $\eps$-close to a unate distribution, $f$ is $(2\eps)$-close to a unate Boolean function. Assuming that $\calD_f$ is $\eps$-close (in total variation distance) to a unate distribution  $\calD'$, there exists a unate function $g': \bits^n \to [0, 1]$ (namely, the PMF of $\calD'$) such that
    \[
        \|g - g'\|_1 \coloneqq \sum_{x \in \bits^n}|g(x) - g'(x)| \le 2\eps.
    \]
    The factor of $2$ appears since the TV distance is half the $L_1$ distance between the PMFs.

    In the rest of the proof, we will show that there exists $g'': \bits^n \to \{0, \alpha\}$ such that $\|g - g''\|_1 \le 2\eps$. Note that the lemma would directly follow from this claim: Since both $g$ and $g''$ are $\{0, \alpha\}$-valued, we have
    \[
        2\eps
    \ge \|g - g''\|_1
    =   \sum_{x \in \bits^n}|g(x) - g''(x)|
    =   \alpha \cdot \sum_{x \in \bits^n}\1{g(x) \ne g''(x)},
    \]
    which implies that $g$ and $g''$ differ on at most $2\eps / \alpha = 2\eps|f^{-1}(1)| \le 2\eps \cdot 2^n$ inputs. Therefore, $g'' / \alpha$ is a unate Boolean function that disagrees with $f = g / \alpha$ on at most a $(2\eps)$-fraction of inputs. In other words, $f$ is $(2\eps)$-close to unate.

    \paragraph{Existence of $g''$.} To show the existence of $g'': \bits^n \to \{0, \alpha\}$ such that $\|g - g''\|_1 \le 2\eps$, we consider a standard linear program for finding the unate function that is the closest to $g$ in $L_1$ distance.
    
    Recall that $g$ is $(2\eps)$-close to a unate function $g'$. Let $E \subseteq \bits^n \times \bits^n$ be the set of all ordered pairs $(x, x')$ such that the constraint $g'(x) \le g'(x')$ is enforced by the unateness of $g'$. Formally, for each edge $\{x^{i \to -1}, x^{i \to +1}\}$ of the hypercube (where $x \in \bits^n$ and $i \in [n]$), we have $(x^{i \to -1}, x^{i \to +1}) \in E$ if $g'$ is monotone in the $i$-th coordinate and $(x^{i \to +1}, x^{i \to -1}) \in E$ if $g'$ is anti-monotone in the $i$-th coordinate. Then, we consider the following linear program (LP) on $2^{n+1}$ real-valued variables $\{h(x): x \in \bits^n\}$ and $\{\delta(x): x \in \bits^n\}$:
    \begin{align*}
        \text{Minimize}&\quad \sum_{x \in \bits^n}\delta(x) &\\
        \text{Subject to}&\quad h(x) \le h(x') &\quad \forall (x, x') \in E,\\
        &\quad -\delta(x) \le g(x) - h(x) \le \delta(x) &\quad \forall x \in \bits^n.
    \end{align*}
    This LP finds the unate function $h$ (with the same monotonicity as $g'$ on all the $n$ coordinates) that is the closest to $g$ in $L_1$ distance: the first set of constraints enforces the unateness of $h$, while the second ensures that the optimal choice of each $\delta(x)$ is $|g(x) - h(x)|$.

    Since $g'$ is a feasible solution to the LP above with an objective value of $\le 2\eps$, by the fundamental theorem of linear programming (e.g., \cite[Theorem 29.13]{CLRS}), there exists a basic feasible solution $(h^*, \delta^*)$ that attains an objective value of $\le 2\eps$. In other words, $h^*$ is a unate function that is $(2\eps)$-close to $g$, and $(h^*, \delta^*)$ is uniquely determined by the linear system obtained from replacing a subset of the constraints with equations.
    
    Then, equations of form ``$h^*(x) = h^*(x')$'' (where $(x, x') \in E$) partition $\bits^n$ into several connected components: two different points $x, x' \in \bits^n$ are in the same component if there exist a path $x = x_1, x_2, \ldots, x_l = x'$ such that the linear system contains the equation $h^*(x_i) = h^*(x_{i+1})$ for every pair of adjacent points $x_i$ and $x_{i+1}$. For each connected component $C \subseteq \bits^n$, there must exist some $x_0 \in C$ such that both equations ``$-\delta^*(x_0) = g(x_0) - h^*(x_0)$'' and ``$g(x_0) - h^*(x_0) = \delta^*(x_0)$'' are in the linear system; otherwise, we may freely set the values $\{h^*(x): x \in C\}$ to the same, arbitrary value, and then set each $\delta^*(x)$ according to ``$-\delta^*(x) = g(x) - h^*(x)$'' or ``$g(x) - h^*(x) = \delta^*(x)$'', depending on which of the two is in the linear system.

    Therefore, for each connected component $C \subseteq \bits^n$, we can find $x_0 \in C$ such that the linear system enforces $-\delta^*(x_0) = g(x_0) - h^*(x_0) = \delta^*(x_0)$. This implies that all variables $\{h^*(x): x \in C\}$ are equal to $g(x_0) \in \{0, \alpha\}$. It then follows that $h^*$ is $\{0, \alpha\}$-valued. Since $h^*$ is unate and $(2\eps)$-close to $g$, it gives the desired function $g''$.
\end{proof}

\begin{lemma}
\label{lemma:function-to-dist-reduction}
    Suppose that algorithm $\calA$ tests whether an unknown distribution $\calD$ over $\bits^n$ is unate or $\eps$-far from unate by drawing $Q_1$ samples from $\calD$ and making $Q_2$ (possibly adaptive) subcube queries on subcubes of dimensions $\le d$. Then, there is an algorithm that tests whether an unknown function $f: \bits^n \to \zo$ is unate or $(2\eps)$-far from unate with an expected query complexity of
    \[
        O\left(\frac{Q_1 + 1}{\eps} + Q_2 \cdot 2^d\right).
    \]
\end{lemma}

\begin{proof}
    We construct a unateness tester $\calA'$ for an unknown function $f: \bits^n \to \zo$ by simulating the distribution tester $\calA$ on distribution $\calD_f$:
    \begin{itemize}
        \item Sample $m = \Theta(1/\eps)$ inputs $\bx^{(1)}, \bx^{(2)}, \ldots, \bx^{(m)} \in \bits^n$ independently and uniformly at random. Query $f$ on these $m$ inputs and compute the empirical average $\hat\mu = \frac{1}{m}\sum_{i=1}^{m}f(\bx^{(i)})$. If $\hat\mu \le \eps$, output ``Accept'' and terminate.
        \item Simulate tester $\calA$ on distribution $\calD_f$. Whenever $\calA$ requests an unconditional sample from $\calD_f$, keep sampling uniformly random inputs until an $x \in \bits^n$ with $f(x) = 1$ is encountered. Send $x$ to $\calA$.
        \item Whenever $\calA$ makes a subcube query on $C \subseteq \bits^n$, query $f$ at every $x \in C$ and return an $x$ sampled from $\{x \in C: f(x) = 1\}$ uniformly at random.
        \item When $\calA$ terminates, make the same decision as $\calA$.
    \end{itemize}

    Let $\calE$ be the event that the following two conditions hold:
    \begin{itemize}
        \item If $f$ has a mean of $\ex{\bx \sim \bits^n}{f(\bx)} \le \eps / 2$, $\hat\mu \le \eps$ holds in the first step of $\calA'$.
        \item If $\ex{\bx \sim \bits^n}{f(\bx)} \ge 2\eps$, $\hat\mu > \eps$ holds in the first step of $\calA'$.
    \end{itemize}
    By a Chernoff bound, we have $\pr{}{\calE} \ge 1 - \delta$ if we set $m = \Theta(\log(1/\delta)/\eps)$.
    
    In the following, we analyze the query complexity, completeness, and soundness of $\calA'$ conditioning on event $\calE$. We will remove the conditioning at the end of the proof.

    \paragraph{Query complexity.} Tester $\calA'$ makes $O(1/\eps)$ queries in the first step. Conditioning on event $\calE$, if $\calA'$ does not terminate in this step, it holds that $\ex{\bx \sim \bits^n}{f(\bx)} > \eps/2$. Then, whenever $\calA$ requests an unconditional sample from $\calD$, $\calA'$ makes $1 / \ex{\bx \sim \bits^n}{f(\bx)} = O(1/\eps)$ queries to $f$ in expectation. Whenever $\calA$ requests a conditional sample from subcube $C \subseteq \bits^n$, $\calA'$ makes $|C| \le 2^d$ queries. Therefore, the expected query complexity of $\calA'$ conditioning on event $\calE$ is
    \[
        O(1/\eps) + Q_1 \cdot O(1/\eps) + Q_2 \cdot 2^d
    =   O\left(\frac{Q_1 + 1}{\eps} + Q_2 \cdot 2^d\right).
    \]

    \paragraph{Completeness and soundness.} Suppose that $f$ is a unate function. By the first part of \Cref{lemma:unate-function-to-dist}, $\calD_f$ is a unate distribution. From the perspective of tester $\calA$ (simulated by $\calA'$), it tests the unateness of $\calD_f$ and thus accepts with probability at least $2/3$. It follows that $\calA'$ also accepts with probability at least $2/3$ (conditioning on event $\calE$) in this case.

    Suppose that $f$ is $(2\eps)$-far from unate. In particular, $f$ is $(2\eps)$-far from the zero function, i.e., it holds that $\ex{\bx \sim \bits^n}{f(\bx)} \ge 2\eps$. Conditioning on event $\calE$, $\hat\mu > \eps$ holds in the first step of $\calA'$, and the tester would not accept by mistake. By the second part of \Cref{lemma:unate-function-to-dist}, distribution $\calD_f$ is $\eps$-far from unate, so the soundness of tester $\calA$ guarantees that $\calA'$ rejects with a conditional probability $\ge 2/3$ in this case.

    \paragraph{Put everything together.} We have showed that, conditioning on event $\calE$ that happens with probability $\ge 1 - \delta$, $\calA'$ has an expected query complexity of $Q \coloneqq O\left(\frac{Q_1 + 1}{\eps} + Q_2 \cdot 2^d\right)$ and succeeds with probability at least $2/3$. Now, suppose that we terminate $\calA'$ whenever it makes more than $Q / \delta$ queries. By Markov's inequality, we obtain a tester with a query complexity always bounded by $Q / \delta$ and an unconditional success probability of at least
    \[
        \pr{}{\calE} \cdot \left(\frac{2}{3} - \delta\right)
    \ge (1 - \delta) \cdot \left(\frac{2}{3} - \delta\right)
    \ge \frac{2}{3} - 2\delta.
    \]

    If we set $\delta = 1/100$, we have $2/3 - 2\delta > 1/2$. Thus, by repeating the truncated tester a constant number of times and taking the majority vote, we obtain a tester with a success probability $\ge 2/3$ and the desired query complexity of
    \[
        O(Q / \delta)
    =   O\left(\frac{Q_1 + 1}{\eps} + Q_2 \cdot 2^d\right). \qedhere
    \]
\end{proof}

%% file: allrefs.bib
@STRING{acm = {ACM Press}}

@string{focs = afocs}

@STRING{ieee = {IEEE Computer Society Press}}

@string{sigact = "SIGACT News"}

@string{stoc = astoc}

@article{chen2025boolean,
  title={Boolean function monotonicity testing requires (almost) $n^{1/2}$ queries},
  author={Chen, Mark and Chen, Xi and Cui, Hao and Pires, William and Stockwell, Jonah},
  eprinttype={arxiv},
  eprint={2511.04558},
  year={2025}
}

@inproceedings{CCRSW25,
  title={Monotonicity Testing of High-Dimensional Distributions with Subcube Conditioning},
  author={Chakrabarty, Deeparnab and Chen, Xi and Ristic, Simeon and Seshadhri, C. and Waingarten, Erik},
  booktitle={Symposium on Theory of Computing (STOC)},
  pages={1019--1030},
  year={2025}
}

@article{BshoutyTamon:96,
  author={N. Bshouty and C. Tamon},
  title={On the {F}ourier spectrum of monotone functions},
  journal={Journal of the ACM},
  year={1996},
  volume={43},
  number={4},
  pages={747-770}
}

@article{RubinfeldS09,
  author    = {Ronitt Rubinfeld and
               Rocco A. Servedio},
  title     = {Testing monotone high-dimensional distributions},
  journal   = {Random Struct. Algorithms},
  volume    = {34},
  number    = {1},
  pages     = {24--44},
  year      = {2009}
  }

@article{Fischer,
author = {E. Fischer},
title = {The art of uninformed decisions: A primer to property testing},
journal = {Bulletin of the European Association for Theoretical Computer
Science},
volume = {75},
pages = {97--126},
year = {2001},
}

@article{GGLRS,
author = {O. Goldreich and S. Goldwasser and E. Lehman and D. Ron
and A. Samordinsky},
title = {Testing Monotonicity},
journal = {Combinatorica},
volume = 20,
number = 3,
pages = {301--337},
year = 2000,
}

@article{Bentkus:03,
  author = {V. Bentkus},
      title = {{On the dependence of the Berry-Esseen bound on
          dimension}},
                journal = "Journal of Statistical Planning and Inference",
                        volume = 113,
                                  pages = {385--402},
                                              year = 2003
                                                          }

@article{Talagrand:93,
  author = {M. Talagrand},
  title = {Isoperimetry, logarithmic {S}obolev inequalities on the discrete cube and {M}argulis' graph connectivity theorem},
  journal = {GAFA},
  volume = 3,
  number = 3,
  year = 1993,
  pages = {298--314}
}

@inproceedings{CWX17stoc,
  author    = {Xi Chen and
               Erik Waingarten and
               Jinyu Xie},
  title     = {Beyond Talagrand functions: new lower bounds for testing monotonicity and unateness},
  booktitle = {Proceedings of the 49th Annual {ACM} {SIGACT} Symposium on Theory
               of Computing (STOC)},
  pages     = {523--536},
  year      = {2017}
}

@inproceedings{CWX17focs,
  author    = {Xi Chen and
               Erik Waingarten and
               Jinyu Xie},
  title     = {Boolean Unateness Testing with $\tilde{O}(n^{3/4})$ Adaptive Queries},
  booktitle = {Proceedings of the 58th Annual {IEEE} Symposium on Foundations of Computer Science (FOCS)},
  pages     = {868--879},
  year      = {2017}
}

@inproceedings{adamaszek2010testing,
  title={Testing monotone continuous distributions on high-dimensional real cubes},
  author={Adamaszek, Micha{\l} and Czumaj, Artur and Sohler, Christian},
  booktitle={Proceedings of the Twenty-First Annual ACM-SIAM Symposium on Discrete Algorithms},
  pages={56--65},
  year={2010},
  organization={SIAM}
}

@article{baleshzar2020optimal,
  title={Optimal Unateness Testers for Real-Valued Functions: Adaptivity Helps},
  author={Baleshzar, Roksana and Chakrabarty, Deeparnab and Pallavoor, Ramesh Krishnan S and Raskhodnikova, Sofya and Seshadhri, C},
  journal={Theory of Computing},
  volume={16},
  number={1},
  pages={1--36},
  year={2020},
  publisher={Theory of Computing Exchange}
}

@inproceedings{khot2016n,
  title={An $o(n)$ queries adaptive tester for unateness},
  author={Khot, Subhash and Shinkar, Igor},
  booktitle={Approximation, Randomization, and Combinatorial Optimization. Algorithms and Techniques (APPROX/RANDOM 2016)},
  pages={37--1},
  year={2016},
  organization={Schloss Dagstuhl--Leibniz-Zentrum f{\"u}r Informatik}
}

@unpublished{chakrabarty2016widetilde,
  title={A ${O}(n)$ non-adaptive tester for unateness},
  author={Chakrabarty, Deeparnab and Seshadhri, C},
  eprint={1608.06980},
  eprinttype={arxiv},
  year={2016}
}

@article{daskalakis2019testing,
  title={Testing ising models},
  author={Daskalakis, Constantinos and Dikkala, Nishanth and Kamath, Gautam},
  journal={IEEE Transactions on Information Theory},
  volume={65},
  number={11},
  pages={6829--6852},
  year={2019},
  publisher={IEEE}
}

@article{canonne2020survey,
  title={A survey on distribution testing: Your data is big. But is it blue?},
  author={Canonne, Cl{\'e}ment L},
  journal={Theory of Computing},
  pages={1--100},
  year={2020},
  publisher={Theory of Computing Exchange}
}

@article{canonne2015testing,
  title={Testing probability distributions using conditional samples},
  author={Canonne, Cl{\'e}ment L and Ron, Dana and Servedio, Rocco A},
  journal={SIAM Journal on Computing},
  volume={44},
  number={3},
  pages={540--616},
  year={2015},
  publisher={SIAM}
}

@article{bhattacharyya2018property,
  title={Property testing of joint distributions using conditional samples},
  author={Bhattacharyya, Rishiraj and Chakraborty, Sourav},
  journal={ACM Transactions on Computation Theory (TOCT)},
  volume={10},
  number={4},
  pages={1--20},
  year={2018},
  publisher={ACM New York, NY, USA}
}

@inproceedings{de2024detecting,
  title={Detecting low-degree truncation},
  author={De, Anindya and Li, Huan and Nadimpalli, Shivam and Servedio, Rocco A},
  booktitle={Proceedings of the 56th Annual ACM Symposium on Theory of Computing},
  pages={1027--1038},
  year={2024}
}

@unpublished{he2023testing,
  title={Testing junta truncation},
  author={He, William and Nadimpalli, Shivam},
  eprint={2308.13992},
  eprinttype={arxiv},
  year={2023},
}

@ARTICLE{de2025testing,
  title     = {Testing convex truncation},
  author    = {De, Anindya and Nadimpalli, Shivam and Servedio, Rocco A.},
  journal   = {Mathematical Statistics and Learning},
  year      = {2025},
  doi 		= {10.4171/MSL/50},
  note      = {Preliminary version in \emph{Proceedings of the 2023 ACM-SIAM Symposium on Discrete Algorithms (SODA)}.},
}

@inproceedings{chen2024uniformity,
  title={Uniformity testing over hypergrids with subcube conditioning},
  author={Chen, Xi and Marcussen, Cassandra},
  booktitle={Proceedings of the 2024 Annual ACM-SIAM Symposium on Discrete Algorithms (SODA)},
  pages={4338--4370},
  year={2024},
  organization={SIAM}
}

@inproceedings{chen2021learning,
  title={Learning and testing junta distributions with sub cube conditioning},
  author={Chen, Xi and Jayaram, Rajesh and Levi, Amit and Waingarten, Erik},
  booktitle={Conference on Learning Theory},
  pages={1060--1113},
  year={2021},
  organization={PMLR}
}

@inproceedings{canonne2021random,
  title={Random restrictions of high dimensional distributions and uniformity testing with subcube conditioning},
  author={Canonne, Cl{\'e}ment L and Chen, Xi and Kamath, Gautam and Levi, Amit and Waingarten, Erik},
  booktitle={Proceedings of the 2021 ACM-SIAM Symposium on Discrete Algorithms (SODA)},
  pages={321--336},
  year={2021},
  organization={SIAM}
}

@inproceedings{blanc2023lifting,
  title={Lifting uniform learners via distributional decomposition},
  author={Blanc, Guy and Lange, Jane and Malik, Ali and Tan, Li-Yang},
  booktitle={Proceedings of the 55th Annual ACM Symposium on Theory of Computing},
  pages={1755--1767},
  year={2023}
}

@article{meel2020testing,
  title={On testing of samplers},
  author={Meel, Kuldeep S and Pote, Yash Pralhad and Chakraborty, Sourav},
  journal={Advances in Neural Information Processing Systems},
  volume={33},
  pages={5753--5763},
  year={2020}
}

@article{pote2022scalable,
  title={On scalable testing of samplers},
  author={Pote, Yash and Meel, Kuldeep S},
  journal={Advances in Neural Information Processing Systems},
  volume={35},
  pages={28068--28079},
  year={2022}
}

@unpublished{adar2024improved,
  title={Improved bounds for high-dimensional equivalence and product testing using subcube queries},
  author={Adar, Tomer and Fischer, Eldar and Levi, Amit},
  eprint={2408.02347},
  eprinttype={arxiv},
  year={2024}
}

@article{blanca2024complexity,
  title={Complexity of high-dimensional identity testing with coordinate conditional sampling},
  author={Blanca, Antonio and Chen, Zongchen and {\v{S}}tefankovi{\v{c}}, Daniel and Vigoda, Eric},
  journal={ACM Transactions on Algorithms},
  volume={21},
  number={1},
  pages={1--58},
  year={2024},
  publisher={ACM New York, NY}
}

@inproceedings{aliakbarpour2019towards,
  title={Towards Testing Monotonicity of Distributions Over General Posets},
  author={Aliakbarpour, Maryam and Gouleakis, Themis and Peebles, John and Rubinfeld, Ronitt and Yodpinyanee, Anak},
  booktitle={Conference on Learning Theory},
  pages={34--82},
  year={2019},
  organization={PMLR}
}

@article{rubinfeld2020monotone,
  title={Monotone probability distributions over the Boolean cube can be learned with sublinear samples},
  author={Rubinfeld, Ronitt and Vasilyan, Arsen},
  journal={11th Innovations in Theoretical Computer Science (ITCS 2020)},
  year={2020}
}

@inproceedings{narayanan2021tolerant,
  title={On tolerant distribution testing in the conditional sampling model},
  author={Narayanan, Shyam},
  booktitle={Proceedings of the 2021 ACM-SIAM Symposium on Discrete Algorithms (SODA)},
  pages={357--373},
  year={2021},
  organization={SIAM}
}

@article{chakraborty2016power,
  title={On the Power of Conditional Samples in Distribution Testing},
  author={Chakraborty, Sourav and Fischer, Eldar and Goldhirsh, Yonatan and Matsliah, Arie},
  journal={SIAM Journal on Computing},
  volume={45},
  number={4},
  pages={1261--1296},
  year={2016},
  publisher={Society for Industrial and Applied Mathematics}
}

@book{CLRS,
  title={Introduction to algorithms},
  author={Cormen, Thomas H. and Leiserson, Charles E. and Rivest, Ronald L. and Stein, Clifford},
  year={2009},
  edition={3rd},
  publisher={MIT press}
}

@article{owen1956tables,
  title={Tables for computing bivariate normal probabilities},
  author={Owen, Donald B},
  journal={The Annals of Mathematical Statistics},
  volume={27},
  number={4},
  pages={1075--1090},
  year={1956},
  publisher={JSTOR}
}

@inproceedings{chakrabarty2013n,
  title={A $o(n)$ monotonicity tester for boolean functions over the hypercube},
  author={Chakrabarty, Deeparnab and Seshadhri, C},
  booktitle={Proceedings of the Forty-Fifth Annual ACM Symposium on Theory of Computing},
  pages={411--418},
  year={2013}
}

@inproceedings{chen2014new,
  title={New algorithms and lower bounds for monotonicity testing},
  author={Chen, Xi and Servedio, Rocco A and Tan, Li-Yang},
  booktitle={2014 IEEE 55th Annual Symposium on Foundations of Computer Science},
  pages={286--295},
  year={2014},
  organization={IEEE}
}

@article{khot2018monotonicity,
  title={On monotonicity testing and boolean isoperimetric-type theorems},
  author={Khot, Subhash and Minzer, Dor and Safra, Muli},
  journal={SIAM Journal on Computing},
  volume={47},
  number={6},
  pages={2238--2276},
  year={2018},
  publisher={SIAM}
}

@inproceedings{lange2022properly,
  title={Properly learning monotone functions via local correction},
  author={Lange, Jane and Rubinfeld, Ronitt and Vasilyan, Arsen},
  booktitle={2022 IEEE 63rd Annual Symposium on Foundations of Computer Science (FOCS)},
  pages={75--86},
  year={2022},
  organization={IEEE}
}

@article{lange2025agnostic,
  title={Agnostic proper learning of monotone functions: beyond the black-box correction barrier},
  author={Lange, Jane and Vasilyan, Arsen},
  journal={SIAM Journal on Computing},
  number={0},
  pages={FOCS23--1},
  year={2025},
  publisher={SIAM}
}

@inproceedings{chen2024mildly,
  title={Mildly exponential lower bounds on tolerant testers for monotonicity, unateness, and juntas},
  author={Chen, Xi and De, Anindya and Li, Yuhao and Nadimpalli, Shivam and Servedio, Rocco A},
  booktitle={Proceedings of the 2024 Annual ACM-SIAM Symposium on Discrete Algorithms (SODA)},
  pages={4321--4337},
  year={2024},
  organization={SIAM}
}

@inproceedings{qiao2018learning,
  title={Learning Discrete Distributions from Untrusted Batches},
  author={Qiao, Mingda and Valiant, Gregory},
  booktitle={9th Innovations in Theoretical Computer Science Conference (ITCS 2018)},
  pages={47--1},
  year={2018},
  organization={Schloss Dagstuhl--Leibniz-Zentrum f{\"u}r Informatik}
}

@article{black2025monotonicity,
  title={A $d^{1/2+o(1)}$ Monotonicity Tester for Boolean Functions on $d$-Dimensional Hypergrids},
  author={Black, Hadley and Chakrabarty, Deeparnab and Seshadhri, C},
  journal={SIAM Journal on Computing},
  number={0},
  pages={FOCS23--147},
  year={2025},
  publisher={SIAM}
}

@inproceedings{chen2019testing,
  title={Testing unateness nearly optimally},
  author={Chen, Xi and Waingarten, Erik},
  booktitle={Proceedings of the 51st Annual ACM SIGACT Symposium on Theory of Computing},
  pages={547--558},
  year={2019}
}

@inproceedings{bhattacharyya2011testing,
  title={Testing monotonicity of distributions over general partial orders.},
  author={Bhattacharyya, Arnab and Fischer, Eldar and Rubinfeld, Ronitt and Valiant, Paul},
  booktitle={ICS},
  pages={239--252},
  year={2011}
}

@inproceedings{batu2004sublinear,
  title={Sublinear algorithms for testing monotone and unimodal distributions},
  author={Batu, Tugkan and Kumar, Ravi and Rubinfeld, Ronitt},
  booktitle={Proceedings of the thirty-sixth annual ACM symposium on Theory of computing},
  pages={381--390},
  year={2004}
}

@article{canonne2022topics,
  title={Topics and techniques in distribution testing: A biased but representative sample},
  author={Canonne, Cl{\'e}ment L and others},
  journal={Foundations and Trends{\textregistered} in Communications and Information Theory},
  volume={19},
  number={6},
  pages={1032--1198},
  year={2022},
  publisher={Now Publishers, Inc.}
}

@article{bezakova2020lower,
  title={Lower bounds for testing graphical models: Colorings and antiferromagnetic ising models},
  author={Bez{\'a}kov{\'a}, Ivona and Blanca, Antonio and Chen, Zongchen and {\v{S}}tefankovi{\v{c}}, Daniel and Vigoda, Eric},
  journal={Journal of Machine Learning Research},
  volume={21},
  number={25},
  pages={1--62},
  year={2020}
}

@inproceedings{canonne2017testing,
  title={Testing bayesian networks},
  author={Canonne, Cl{\'e}ment L and Diakonikolas, Ilias and Kane, Daniel M and Stewart, Alistair},
  booktitle={Conference on Learning Theory},
  pages={370--448},
  year={2017},
  organization={PMLR}
}

@inproceedings{daskalakis2017square,
  title={Square Hellinger subadditivity for Bayesian networks and its applications to identity testing},
  author={Daskalakis, Constantinos and Pan, Qinxuan},
  booktitle={Conference on Learning Theory},
  pages={697--703},
  year={2017},
  organization={PMLR}
}

@article{canonne2019testing,
  title={Testing $k$-monotonicity: The rise and fall of boolean functions},
  author={Canonne, Cl{\'e}ment L and Grigorescu, Elena and Guo, Siyao and Kumar, Akash and Wimmer, Karl},
  journal={Theory of Computing},
  volume={15},
  number={1},
  pages={1--55},
  year={2019},
  publisher={Theory of Computing Exchange}
}

@article{grigorescu2017k,
  title={Flipping out with many flips: Hardness of testing $k$-monotonicity},
  author={Grigorescu, Elena and Kumar, Akash and Wimmer, Karl},
  journal={SIAM Journal on Discrete Mathematics},
  volume={33},
  number={4},
  pages={2111--2125},
  year={2019},
  publisher={SIAM}
}

@article{chen2022new,
  title={New Distinguishers for Negation-Limited Weak Pseudorandom Functions},
  author={Chen, Zhihuai and Guo, Siyao and Li, Qian and Lin, Chengyu and Sun, Xiaoming},
  eprint={2203.12246},
  eprinttype={arxiv},
  year={2022}
}

@article{acharya2015optimal,
  title={Optimal testing for properties of distributions},
  author={Acharya, Jayadev and Daskalakis, Constantinos and Kamath, Gautam},
  journal={Advances in Neural Information Processing Systems},
  volume={28},
  year={2015}
}
